\documentclass[conference]{IEEEtran}

\usepackage{balance}
\usepackage{colortbl}
\usepackage{tabulary}
\usepackage{booktabs}
\usepackage{subcaption}
\usepackage{setspace}
\usepackage{amsmath}

\usepackage{bbm}
\usepackage{todonotes}
\usepackage{epsf}
\usepackage{epsfig}
\usepackage{times}
\usepackage{epsfig}
\usepackage{graphicx}
\usepackage{bbold}
\usepackage{mathrsfs}
\usepackage{amssymb}
\usepackage{pdfpages}
\usepackage{epstopdf}
\usepackage{tcolorbox,bbold}
\usepackage{authblk}
\usepackage{algorithm,algorithmicx,algpseudocode}
\makeatletter
\newcommand{\algmargin}{\the\ALG@thistlm}
\makeatother
\algnewcommand{\parState}[1]{\State%
  \parbox[t]{\dimexpr\linewidth-\algmargin}{\strut #1\strut}}
\usepackage{url}
\usepackage{dsfont}
\usepackage{lettrine} 
\usepackage{amsmath,epsfig,amssymb,algorithm,algpseudocode,amsthm,cite,url}
\IEEEoverridecommandlockouts
\allowdisplaybreaks
\usepackage{csquotes}

\usepackage{verbatim}
\usepackage[english]{babel}
\usepackage{amsmath,amssymb}

\usepackage[justification=centering]{caption}
\usepackage{verbatim}

\newtheorem{theorem}{\bf Theorem}

\newtheorem{proposition}{\bf Proposition}

\begin{document}
\clearpage
\title{\huge Ultra-Reliable Indoor Millimeter Wave Communications using Multiple Artificial Intelligence-Powered Intelligent Surfaces}
\author{Mehdi Naderi Soorki$^{1}$, Walid Saad$^{2,4}$, Mehdi Bennis$^{3}$, and Choong Seon Hong$^{4}$\vspace*{0.1cm}\\
\small {$^{1}$Engineering Faculty, Shahid Chamran University of Ahvaz, Iran.}\\
\small {$^{2}$Wireless@VT, Bradley Department of Electrical and Computer Engineering, Virginia Tech, Blacksburg, VA, USA.}\\
\small {$^{3}$Centre for Wireless Communications, University of Oulu, Finland.}\\
\small {$^{4}$ Department of Computer Science and Engineering, Kyung Hee University, South Korea.}\\
\small {Emails: m.naderisoorki@scu.ac.ir, walids@vt.edu, mehdi.bennis@oulu.fi, cshong@khu.ac.kr}
\vspace{-0.7cm}
  \thanks{This research was supported by the U.S. National Science Foundation under Grants CNS-1836802 and CNS-2030215. A preliminary version of this work appears in~\cite{Globecome2019}.}
}
\maketitle
\thispagestyle{empty}

\begin{abstract}
In this paper, a novel framework for guaranteeing ultra-reliable millimeter wave (mmW) communications using multiple artificial intelligence (AI)-enabled reconfigurable intelligent surfaces (RISs) is proposed. The use of multiple AI-powered RISs allows changing the propagation direction of the signals transmitted from a mmW access point (AP) thereby improving coverage particularly for non-line-of-sight (NLoS) areas. However, due to the possibility of highly stochastic blockage over mmW links, designing an intelligent controller to jointly optimize the mmW AP beam and RIS phase shifts is a daunting task. In this regard, first, a parametric risk-sensitive episodic return is proposed to maximize the expected bitrate and mitigate the risk of  mmW link blockage. Then, a closed-form approximation of the policy gradient of the risk-sensitive episodic return is analytically derived. Next, the problem of joint beamforming for mmW AP and phase shift control for mmW RISs is modeled as an identical payoff stochastic game within a cooperative multi-agent environment, in which the agents are the mmW AP and the RISs. Two centralized and distributed controllers are proposed to control the policies of the mmW AP and RISs. To \emph{directly} find a near optimal solution, the parametric functional-form policies for these controllers are modeled using deep recurrent neural networks (RNNs). The deep RNN-based controllers are then trained based on the derived closed-form gradient of the risk-sensitive episodic return. It is proved that the gradient updating algorithm converges to the same locally optimal parameters for deep RNN-based centralized and distributed controllers. Simulation results show that the error between policies of the optimal and the RNN-based controllers is less than $1.5\%$. Moreover, the variance of the achievable rates resulting from the deep RNN-based controllers is $60\%$ less than the variance of the risk-averse baseline.
\end{abstract}
{ \emph{Index Terms}--- Millimeter wave networks; RIS; 5G and  beyond; stochastic games; deep risk-sensitive reinforcement learning.}
\vspace{-0.21cm}
\section{Introduction}\label{sec:Intro}
Millimeter wave (mmW) communications is a promising solution to enable high-speed wireless access in 5G wireless networks and beyond~\cite{Walid2019,Petrov2017}. Nevertheless, the high attenuation and scattering of mmW propagation makes guaranteeing the coverage of mmW wireless networks very challenging. To overcome high attenuation and scattering of mmW propagation challenges, integrating massive antennas for highly directional beamforming at both mmW access point (AP) and user equipment (UE) has been proposed~\cite{Walid2019}. However, applying beamforming techniques will render the use of directional mmW links very sensitive to random blockage caused by people and objects in a dense environment. This, in turn, gives rise to unstable line-of-sight (LoS) mmW links and unreliable mmW communications~\cite{Petrov2017}. To provide robust LoS coverage, one proposed solution is to deploy ultra-dense APs and active relay nodes to improve link quality using multi-connectivity for a given UE~\cite{Petrov2017,Abari2018,Journal2019,Chaccour2}. However, the deployment of multiple mmW APs and active relay nodes is not economically feasible and can lead to high control signalling overhead. To decrease signalling overhead and alleviate economic costs while also establishing reliable communications, recently the use of reconfigurable intelligent surfaces (RISs) has been proposed~\cite{Walid2019,Huang2018,Christos2018,Basar2019,Zhang2019,Chaccour1}.\footnote{Note that, when an RIS is used as a passive reflector, it is customary to use the term intelligent reflecting surface (IRS) to indicate this mode of operation. Meanwhile, when an RIS is used as a large surface with an active transmission, the term large intelligent surface (LIS) is commonly used~\cite{Huang2018,Basar2019,Zhang2019}.}
\vspace{-0.1cm}
\subsection{Prior Works}
RISs are man-made surfaces including conventional reflect-arrays, liquid crystal surfaces, and software-defined meta-surfaces that are electronically controlled~\cite{Walid2000}. In a mmW network enabled with RIS, mmW RISs are turned into a software-reconfigurable entity whose operation is optimized to increase the availability of mmW LoS connectivity. Thus, the RISs reflect the mmW signals whenever possible to bypass the blockages. One of the main challenges in using reconfigurable RISs is how to adjust the phases of the reflected waves from different RISs so that the LOS and reflected mmW signals can be added coherently, and the signal strength of their sum is maximized.

The previous works in~\cite{Abari2018,Huang2018,58,91}, and~\cite{102} use conventional optimization techniques for addressing joint transmit and passive beamforming challenges in an RIS-assisted wireless networks. In~\cite{58} and~\cite{91}, the authors minimize the transmit power for an RIS-enhanced MISO system in both single-user and multi-user scenarios. Alternating optimization (AO) based algorithms for passive beamforming were developed to find locally-optimal solutions. This work shows that an RIS can simultaneously enhance the desired signal strength and mitigate the interference for the multi-user scenario. The same problem was further investigated in~\cite{91} by taking discrete RIS phase shifts into consideration. The authors derive the optimal solutions by applying the branch-and-bound method and exhaustive search for single-user and multi-user scenarios, respectively. In~\cite{102}, the authors proposed a new approach to maximize spectral efficiency in an RIS-enhanced downlink MIMO system. The passive beamforming was designed by using the sum of gains maximization principle and by using the alternating direction method of multipliers. Several recent works such as in~\cite{Abari2018} and~\cite{Huang2018} have been proposed to establish reliable mmW links. In~\cite{Huang2018}, the authors present efficient designs for both transmit power allocation and RIS phase shift control. Their goal is to optimize spectrum or energy efficiency subject to individual rate requirements for UEs. However, the work in~\cite{Huang2018} does not consider stochastic blockage over mmW LoS links and, thus, its results cannot be generalized to a real-world mmW system. In~\cite{Abari2018}, the authors implement a smart mmWave reflector to provide high data rates between a virtual reality headset and game consoles. To handle beam alignment and tracking between the mmWave reflector and the headset, their proposed method must try every possible combination of mirror transmit beam angle and headset receive beam angle incurring significant overhead due to the brute-force solution for beam alignment. However, the works in~\cite{Abari2018,Huang2018,58,91}, and~\cite{102} have some limitations. First, they assumed that the users are generally static for simplicity. Moreover, the communication environment is assumed to be perfectly known and the instantaneous channel state information of all the channels are assumed to be available at the BS. Lastly, the RIS/BS are not capable of learning from the unknown environment or from the limited feedback of the users~\cite{9424177}. Hence, the conventional RIS-enhanced wireless networks designed based on the works in~\cite{Abari2018,Huang2018,58,91}, and~\cite{102} can not guarantee ultra-reliable mmW communication in a real environment with highly dynamic changes due to the users mobility and the risk of NLoS mmW link.

In practice, an intelligent solution that uses tools such as  machine learning (ML) is needed for RIS-assisted mmW networks~\cite{Park2019}. In this regards, the works in~\cite{128,129,144,142} and~\cite{Globecome2019} focused on designing ML-empowered RIS-enhanced wireless networks which can adaptively configure the networks  to unknown random changes in the environment. In~\cite{128}, the authors proposed a deep neural network (DNN)-based approach to estimate the mapping between a user's position and the configuration of an RIS in an indoor environment, with the goal of maximizing the received SNR. In~\cite{129}, the authors proposed an algorithm based on deep learning (DL) for optimally designing the RIS phase shift by training the DNN offline. In~\cite{144}, the authors proposed a deep reinforcement learning (RL) algorithm for maximizing the achievable communication rate by directly optimizing interaction matrices from the sampled channel knowledge. In the proposed deep RL model, only one beam was used for each training episode. In\cite{142}, the authors applied a deep deterministic policy gradient (DDPG) algorithm for maximizing the throughput by using the sum rate as instant rewards for training the DDPG model. However, the works in~\cite{128,129,144,142} did not address the challenges of reliability in highly-dynamic mmW networks when using RL. Moreover, for simplicity, these works did not consider the cooperation and coordination between multiple RISs in their systems. Towards this vision, in~\cite{Globecome2019}, an intelligent controller based on DNNs for configuring mmW RISs is studied. The approach proposed in~\cite{Globecome2019} guarantees ultra-reliable mmW communication and captures the unknown stochastic blockages, but it is limited for a scenario with only one RIS. However, simultaneously employing multiple RISs becomes more challenging due to the need for cooperation amongst RISs~\cite{9424177}. Thus, a new framework to guarantee ultra-reliable communication in the ML-empowered RIS-enhanced mmW networks is needed so as to provide robust, stable and near-optimal solution for the coordinated beamforming and phase-shift control policy.
\vspace{-0.2cm}
\subsection{Contributions}
The main contribution of this paper is to propose \emph{a novel framework for guaranteeing ultra-reliable mobile mmW communications using artificial intelligence (AI)-powered RISs}. The proposed approach allows the network to autonomously form transmission beams of the mmW AP and control the phase of the reflected mmW signal in the presence of stochastic blockage over the mmW links. To solve the problem of joint beamforming and phase shift-control in an RISs-assisted mmW network while guaranteeing ultra-reliable mmW communications, we formulate a stochastic optimization problem whose goal is to maximize a parametric risk-sensitive episodic return. The parametric risk-sensitive episodic return not only captures the expected bitrate but is also sensitive to the risk of NLoS mmW link over future time slots. Subsequently, we use deep and risk-sensitive RL to solve the problem in an online manner. Next, we model the risk-sensitive RL problem as an identical payoff stochastic game in a cooperative multi-agent environment in which the agents are mmW AP and RISs~\cite{Walidbook2000}. Two centralized and distributed controllers are proposed to control the policy of the mmW AP and RISs in the identical payoff stochastic game. To find a near optimal solution, the parametric functional-form policies are implemented using a deep RNN~\cite{Walid2017} which \emph{directly} search the optimal policies of the beamforming and phase shift-controllers. In this regard, we analytically derive a closed-form approximation for the gradient of risk-sensitive episodic return, and the RNN-based policies are subsequently trained using this derived closed-form gradient. We prove that if the centralized and distributed controllers start from the same strategy profile in the policy space of the proposed identical payoff stochastic game, then the gradient update algorithm will converge to the same locally optimal solution for deep RNNs. Simulation results show that the error between the policies of the optimal and RNN-based controllers is small. The performance of deep RNN-based centralized and distributed controllers is identical. Moreover, for a high value of risk-sensitive parameter, the variance of the achievable rates resulting from the deep RNN-based controllers is $60\%$ less than the non-risk based solution. The main continuations of this paper are summarized as follows:
\begin{itemize}
\item We propose a novel smart conrol framework based on artificial intelligence for guaranteeing ultra-reliable mobile mmW communications when multiple RISs are used in an indoor scenario. The proposed approach allows the network to autonomously form transmission beams of the mmW AP and control the phase of the reflected mmW signal from mmW RIS in the presence of unknown stochastic blockage. In this regard, we formulate a new joint stochastic beamforming and phase shift-control problem in an RISs-assisted mmW network under ultra-reliable mmW communication constraint. Our objective is to maximize a parametric risk-sensitive episodic return. The parametric risk-sensitive episodic return not only captures the expected bitrate but is also sensitive to the risk of NLoS mmW link over future time slots.

\item We apply both risk-sensitive deep RL and cooperative multi-agent sysem to find a solution for the joint stochastic beamforming and phase shift-control problem, in an online manner. We model the risk-sensitive RL problem as an identical payoff stochastic game in a cooperative multi-agent environment in which the agents are mmW AP and RISs. Then, we propose two centralized and distributed control policies for the transmission beams of mmW AP and phase shift of RISs.

\item To find a near optimal solution for our proposed centralized and distributed control policies,we implement parametric functional-form policies using a deep RNN which can \emph{directly} search the optimal policies of the beamforming and phase shift controllers. We \emph{analytically} derive a closed-form approximation for the gradient of risk-sensitive episodic return, and the RNN-based policies are subsequently trained using this derived closed-form gradient.

\item We mathematically prove that, if the centralized and distributed controllers start from the same strategy profile in the policy space of the proposed identical payoff stochastic game, then the gradient update algorithm will converge to the same locally optimal solution for deep RNNs.  Moreover, we mathematically show that, at the convergence of the gradient update algorithm for the RNN-based policies, the policy profile under the distributed controllers is a Nash equilibrium equilibrium of the RL and cooperative multi-agent system.

\item Simulation results show that the error between the policies of the optimal and RNN-based cotrollers is small. The performance of deep RNN-based centralized and distributed controllers is identical. Moreover, for a high value of risk-sensitive parameter, the variance of the achievable rates resulting from the deep RNN-based controllers is 60\% less than the non-risk based solution.
\end{itemize}

The rest of the paper is organized as follows. Section~\ref{Sec:Sys-Model} presents the system model and the stochastic and risk-sensitive optimization problem in the smart reflector-assisted mmW networks. In Section~\ref{Sec:Algorithm}, based on the framework of deep and risk-sensitive RL, we propose a deep RNN to solve the stochastic and risk-sensitive optimization problem for the optimal reflector configuration. Then, in Section~\ref{Sec:Simulation}, we numerically evaluate the proposed policy-gradient approach. Finally, conclusions are drawn in Section~\ref{Sec:Conclusion}.
\section{System Model and Problem Formulation}\label{Sec:Sys-Model}
\subsection{System model}
Consider the downlink between an UE and an indoor RIS-assisted mmW network composed of one mmW AP and multiple AI-powered RISs. In this network,  due to the blockage of mmW signals, there exist areas where it is not possible to establish LoS mmW links between the mmW AP and UE, particularly for mobile user. We call these areas as \emph{dark areas}. Each mmW AP and UE will have, respectively, $N_a$ and $N_u$ antennas to form their beams. In our model, there are $G$ AI-powered mmW RISs that intelligently reflect the mmW signals from the mmW AP toward the mobile UE located in the dark areas. Each mmW RIS $g$ is in the form of a uniform planar array (UPA) and consists of $N_g=N_{gh}\times N_{gv}$ meta-surfaces where $N_{gh}$ and $N_{gv}$ separately denote the sizes along the horizontal and vertical dimensions. The size of an RIS at mmW bands will be smaller than the size of a typical indoor scenario or the distance between user and RISs which is often in the order of more than 1 meter in an indoor environment. Thus, as mentioned in recent works such as\cite{Abari2018,Kamoda2011} and~\cite{Tan2018} for mmW communication, we can consider far-field characteristics for mmW signals reflected from an RIS. Here, we consider discrete time slots indexed by $t$. The beam angle of the mmW AP directional antenna is represented by $\theta_{0,t}$ at time slot $t$, where index $0$ represents the mmW AP. If the user is in the LoS coverage of the mmW AP, then $\theta_{0,t}$ is matched to the direction of the user toward the mmW AP. In this case, we assume that the AP transmission beam angle $\theta_{0,t}$ is chosen from a set $\Theta=\{-\pi+\frac{2a\pi}{A-1}|a=0,1,...,A-1\}$ of $A$ discrete values. Let the angle between mmW AP and reflector $g$ be $\phi_g$. However, if the user moves to the dark areas, the mmW AP chooses the antenna direction toward one of the mmW RISs $g$, $\theta_{0,t}=\phi_g$. Hence, $\phi_g \in \Theta$. When the mmW RIS $g$ receives a signal from the mmW AP, the mmW RIS establishes a LoS mmW link using a controlled reflected path to cover the user in the dark areas. Hence, the cascaded channel consists of two parts: the link from the AP to the RIS, i.e., $\boldsymbol{H}_{\text{AP-to-RIS},gt} \in \mathbb{C}^{N_a \times N_g}$, and the link from the RIS to the UE, i.e., $\boldsymbol{H}_{\text{RIS-to-UE},gt}^H \in \mathbb{C}^{N_g \times N_u}$, at time slot $t$~\cite{two}. Hence, the cascaded channel over the AP-RIS-UE link will be $\boldsymbol{H}_{\text{AP-to-RIS},gt} \Psi_{gt} \boldsymbol{H}_{\text{RIS-to-UE},gt}$. We consider an ideal RIS with a reflection gain equal to one. Hence, $\Psi_{gt} \in \mathbb{C}^{N_g \times N_g}$ is a diagonal matrix, i.e., $\Psi_{gt}=\text{diag}\{e^{j\psi_{1,t}},...,e^{j\psi_{N_g,t}}\}$. Notice that $\psi_{i,t}$ in $\Psi_{gt}$ represent the phase shift introduced by each RIS meta-surface $i$. Let $\Psi_{gt} \in \Psi=\{ \Psi^{(b)} |b=1,...,B\}$ have $B$ possible values indexed $b$.

Here, we consider a well-known multi-ray mmW channel model~\cite{Shahmansoori2018}. In this channel model for mmW links, there are $L$ rays between mmW transmitter and receiver, and each ray can be blocked by an obstacle. For $L$ mmW rays, let $\theta_{g,t}=\{\theta_{gl,t}|l=1,2,...,L\}$ and $\omega_{g,t}=\{\omega_{gl,t}|l=1,2,...,L\}$  be, respectively, the azimuth and elevation AoD reflection beams from mmW RIS $g$, and $\phi_{t}=\{\phi_{l,t}|l=1,2,...,L\}$ be the AoA beam at UE, then, the channel matrix over the RIS-to-UE mmW path at time slot $t$ will be given by~\cite{Shahmansoori2018}:
\begin{align}
&\boldsymbol{H}_{\text{RIS-to-UE},gt}= [\boldsymbol{b}_{\text{Tx}}(\theta_{g1,t},\omega_{g1,t}),...,\boldsymbol{b}_{\text{Tx}}(\theta_{gL,t},\omega_{gL,t})] \times \nonumber \\ &\text{diag}( \tilde{\boldsymbol{\alpha}}_{gt}(d_{gu})) \times
 [\boldsymbol{a}_{\text{Rx}}(\phi_{1,t}),...,\boldsymbol{a}_{\text{Rx}}(\phi_{L,t})]^H,
 \label{RIS2UE}
\end{align}
where $\boldsymbol{b}_{\text{Tx}}(\theta_{gl,t},\omega_{gl,t})$ and $\boldsymbol{a}_{\text{Rx}}(\phi_{l,t})$ are the spatial steering vectors~\cite{two}.

Here, $\boldsymbol{a}_{\text{Rx}}(\phi_{l,t})=[e^{j\frac{N_u-1}{2}\pi \cos({\phi}_{l})},...,e^{-j\frac{N_u-1}{2}\pi \cos({\phi}_{l})}]^T\in \mathbb{C}^{N_u \times 1}$ denotes the array response vectors for AoA at the UE for ray $l$. Correspondingly, $\boldsymbol{b}_{\text{Tx}}(\theta_{gl,t},\omega_{gl,t})$, which denotes the array response vectors for AoD at the RIS $g$ for ray $l$, can be written as $\boldsymbol{b}_{\text{Tx}}(\omega_{gl,t},\theta_{gl,t})= \boldsymbol{b}_{\text{el}}(\theta_{l,t}) \otimes \boldsymbol{b}_{\text{az}}(\omega_{gl,t},\theta_{gl,t})\in \mathbb{C}^{N_g \times 1}$, where $\boldsymbol{b}_{\text{el}}(\theta_{gl,t})=[e^{j\frac{N_{gv}-1}{2}\pi \cos(\theta_{gl,t})},...,e^{-j\frac{N_{gv}-1}{2}\pi \cos(\theta_{gl,t})}]^T \in \mathbb{C}^{N_{gv} \times 1}$ and $\boldsymbol{b}_{\text{az}}(\omega_{gl,t},\theta_{gl,t})=[e^{j\frac{N_{gh}-1}{2}\pi \cos(\omega_{gl,t}) \sin(\theta_{gl,t})},...,e^{-j\frac{N_{gh}-1}{2}\pi \cos(\omega_{gl,t}) \sin(\theta_{gl,t}) }]^T \in \mathbb{C}^{N_{gh} \times 1}$.
Furthermore, $\otimes$ represents the Kronecker product operator and $[.]^H$ represents the conjugate transpose. Moreover,
$\tilde{\boldsymbol{\alpha}}_{gt}(d_{gu})=[\alpha_{g1}\sqrt{\rho_{g1}},...,\alpha_{gL}\sqrt{\rho_{gL}}]$ where $\rho_{gl}\in \{ (\frac{c}{2\pi f_c})^2 d_{gu}^{-\nu_L}, (\frac{c}{2\pi f_c})^2 d_{gu}^{-\nu_{NLoS}}\}$ and $\alpha_{gl}$ is the complex channel gain of path $l$ from the mmW RIS $g$ to the UE~\cite{Shahmansoori2018}. Here, $d_{gu}$ is the distance between mmW RIS $g$ and UE, and $\nu_{L}$ and $\nu_{NLoS}$ are the slopes of the best linear fit to the propagation measurement in mmW frequency band for LoS and NLoS mmW links, respectively.Similarly, for $L$ mmW rays, let $\upsilon_{g,t}=\{\upsilon_{gl,t}|l=1,2,...,L\}$ and $\phi_{g,t}=\{\phi_{gl,t}|l=1,2,...,L\}$  be respectively the azimuth and elevation AoA beams at RIS $g$, and $\theta_{0,t}=\{\theta_{0l,t}|l=1,2,...,L\}$ be the AoD beam from AP, then, the channel matrix over the AP-to-RIS mmW path at time slot $t$ is given by~\cite{Shahmansoori2018}:
\begin{align}
&\boldsymbol{H}_{\text{AP-to-RIS},t}=
[\boldsymbol{a}_{\text{Tx}}(\theta_{01,t}),...,\boldsymbol{a}_{\text{Tx}}(\theta_{0L,t})] \nonumber \\ &
\times \text{diag}(\tilde{\boldsymbol{\beta}}_{gt}(d_{g}) \times
[\boldsymbol{b}_{\text{Rx}}(\upsilon_{g1,t},\phi_{g1,t}),...,\boldsymbol{b}_{\text{Rx}}(\upsilon_{gL,t},\phi_{gL,t})]^{H},
\label{AP2RIS}
\end{align}
where the structure of $\boldsymbol{b}_{\text{Rx}}(\upsilon_{gl,t},\phi_{gl,t})$ and $\boldsymbol{a}_{\text{Tx}}(\theta_{0l,t})$ are similar to that of $\boldsymbol{b}_{\text{Tx}}(\omega_{gl,t},\theta_{gl,t})$ and $\boldsymbol{a}_{\text{Rx}}(\phi_{l,t})$, respectively. Moreover, the structure of  $\tilde{\boldsymbol{\beta}}_{gt}(d_{g})$ is similar to  $\tilde{\boldsymbol{\alpha}}_{gt}(d_{gu})$
where $d_{g}$ is the distance between the mmW AP and RIS $g$. Consequently, the channel matrix between the mmW AP and UE over one mmW AP-to-UE link and $G$ mmW AP-to-RIS-to-UE links is defined as:
\begin{align}
& \boldsymbol{H}_t(\theta_{0,t},\Psi_{1t},...,\Psi_{Gt})= \boldsymbol{H}_{\text{AP-to-UE},t}+\nonumber \\ &
\sum_{g=1}^G \boldsymbol{H}_{\text{AP-to-RIS},gt} \Psi_{gt} \boldsymbol{H}_{\text{RIS-to-UE},gt},
\label{H_total}
\end{align}
where $\boldsymbol{H}_{\text{AP-to-UE},t}=[\boldsymbol{a}_{\text{Tx}}(\theta_{01,t}),...,\boldsymbol{a}_{\text{Tx}}(\theta_{0L,t})] \times \text{diag}(\tilde{\boldsymbol{\gamma}}_{t}(d_u)) \times [\boldsymbol{a}_{\text{Rx}}(\phi_{1,t}),...,\boldsymbol{a}_{\text{Rx}}(\phi_{L,t})]^H$ is the channel matrix over the mmW AP-to-UE link. Here, the structure of  $\tilde{\boldsymbol{\gamma}}_{t}(d_{u})$ is similar to  $\tilde{\boldsymbol{\alpha}}_{gt}(d_{gu})$ where $d_{u}$ is the distance between mmW AP and UE.

In our model, there are two links: transmission link and control link. The transmission link uses the mmW band to send data over AP-to-UE or AP-to-RIS-to-UE links. The sub-6 GHz link is only used to transmit control signals to the controllers of mmW AP and RISs.  At the beginning of the downlink transmission, since the exact location of the UE is not known to the controller, we apply the three step low-complexity beam search algorithm presented in~\cite{ChannelEstimation} in our model to find the angle-of-departure of transmission beam from mmW AP $\theta_{0,t}$ for $t=0$, and complex path gains, either LoS or NLoS links to the UE. As a result, at the beginning of transmission, if an LoS complex path gain is found, the mmW AP forms its beam directly toward the UE. But, if a NLoS complex path gain is found, the mmW AP sequentially forms its beam toward the mmW reflectors and the beam search algorithm will be applied again. In this case, each mmW reflector changes their reflection angle to sweep all the dark area, until the LoS path loss gain between mmW reflector and UE, and initial relection angles of RISs, $\theta_{g,t}$ for $t=0$, are detected. However in the future time slots, the availability of LoS link as well as the channel gain are random variables with unknown distributions due to the mobility of user, and the AoA signals $\tilde{\phi}_{t}$ at the UE  is a stochastic variable which randomly changes due to unknown factors such as the user's orientation. In our model, the channel can be considered as a multi-ray channel between mmW AP as transmitter node and UE as a receiver node or a MIMO channel between $N_a$ transmitter antenna and $N_u$ receiver antenna where there are $G$ reflectors as the scatterers between them. Thus, following the multi-ray channel model in (\ref{H_total}) and the received bitrate in a RIS-assisted MIMO networks~\cite{9403420}, the total achievable bitrate is given by:
\begin{equation}
r_{t}(\theta_{0,t},\Psi_{1t},...,\Psi_{Gt})=w \log_2 \text{det}\left[\boldsymbol{I}_{N_a}+\frac{q}{N_a w\sigma^2}
\boldsymbol{H}_t \boldsymbol{H}_t^H\right],
\label{Bit_rate}
\end{equation}
where $q$ is the transmission power, $w$ is the mmW bandwidth, and $\sigma^2$ is the noise density.

Fig.~\ref{System_model_1} is an illustrative example that shows how one mmW AP and two mmW RISs are used to bypass the blockage during four consecutive time slots $t$ and $t+3$. As seen in Fig.~\ref{System_model_1}, since the user is in the dark area during time slots $t$ and $t+2$ for mmW AP, the mmW RISs are used to provide coverage for the user. Here, during two time slots $t$ and $t+1$, the mmW AP transmits the signal toward the reflector $1$, and then this reflector reflects the signal toward the user moving in the dark area $1$. Thus, the beam angles of mmW AP signals are $\theta_{0,t}=\phi_1$ and $\theta_{0,t+1}=\phi_{1}$ at time slots $t$ and $t+1$. Then, since the user is moving in the dark area $2$ during time slot $t+2$, the mmW AP transmits the signal toward the mmW RIS $2$, $\theta_{0,t+2}=\phi_2$. In this case, the user is not in the LoS coverage of reflector $1$ and the reflector $2$ reflects the signal toward $\theta_{2,t+2}$ to cover the user at time slot $t+2$. As shown in Fig.~\ref{System_model_1}, the user is not in any dark area at time slot $t+3$ and mmW AP can directly transmit the signal over LoS link toward the user, $\theta_{0,t+3}$. The list of our main notations is given in Table~\ref{tab:Symbols}.
\begin{figure}[!t]
	\begin{center}
		\includegraphics[width=8.5cm]{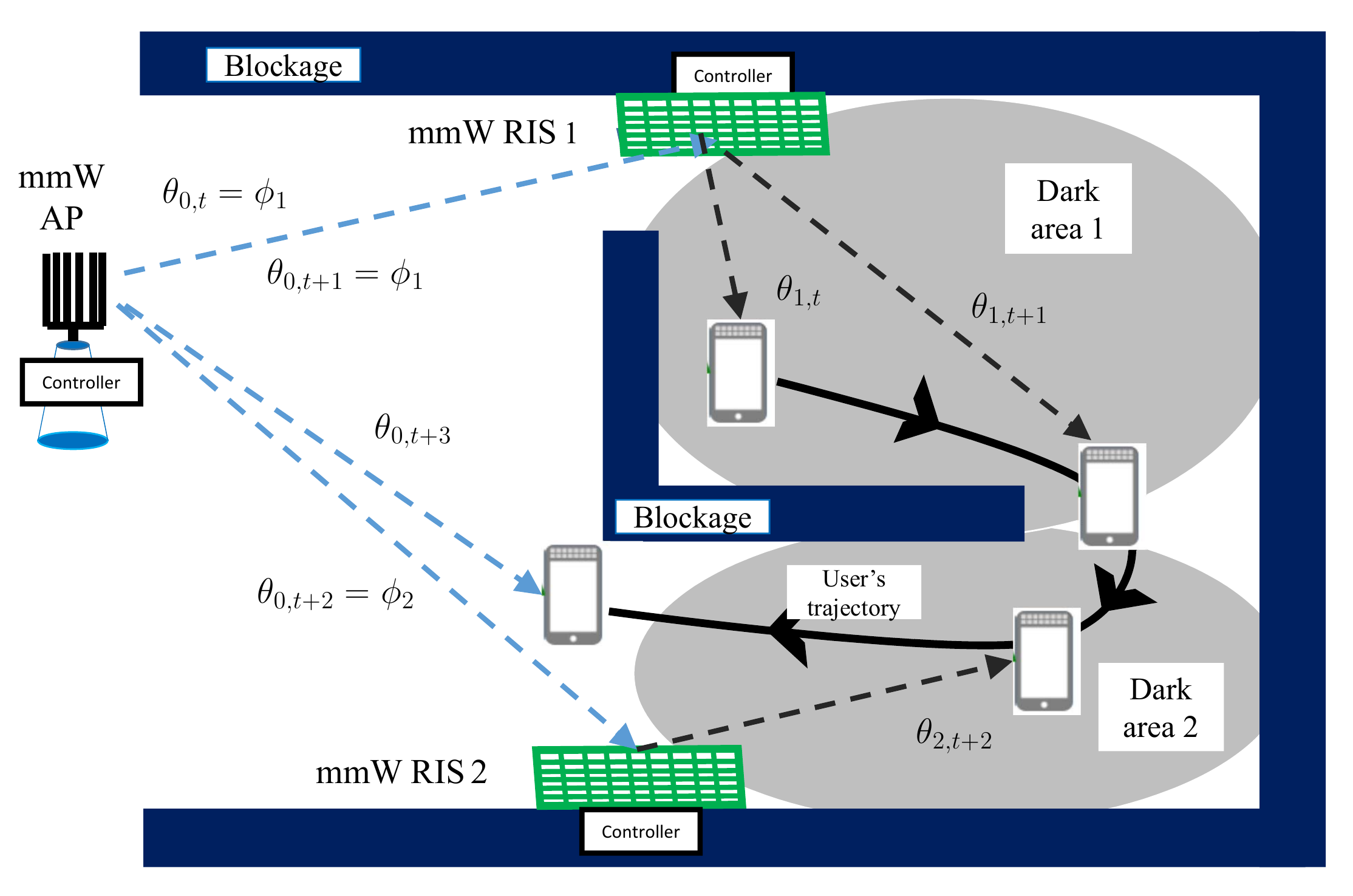}		\vspace{-0.1cm}
		\caption{ \small An illustrative example of the system model with one mmW AP and two mmW RISs.}
        \vspace{-0.4cm}
		\label{System_model_1}
	\end{center}
\end{figure}

As a result, the phase shift-control policy must not only consider the unknown future trajectory of mobile users but also adapt itself to the possible stochastic blockages in future time slots. In this case, an intelligent policy for the phase shift-controller, which can predict unknown stochastic blockages, is required for mmW AP and mmW RISs to guarantee ultra-reliable mmW communication, particularly for indoor scenarios with many dark areas.
\begin{table}[t]
\caption{List of our notations.}
\small
\vspace{-.5cm}
\begin{center}
\begin{tabular}{l | l } \toprule
{ \textbf{Symbol}}  & { \textbf{Definition}}   \\  \rowcolor[gray]{.9} \hline
$N_a$   &  Number of mmW AP antennas\\
$N_u$   &  Number of mmW AP antennas\\ \rowcolor[gray]{.9}
$N_g$   & Number of meta-surfaces per mmW RIS\\
$\theta_{0,t}$  & The beam angle of AP at timeslot $t$\\ \rowcolor[gray]{.9}
$\Psi_{gt}$ & The phase shift introduced by RIS $g$ at times lot $t$\\
$\phi_{g}$  & The AoA at RIS $g$\\ \rowcolor[gray]{.9}
$\tilde{\phi}_{t}$  & The stochastic AoA at UE at timeslot $t$\\
$\boldsymbol{H}_{0,t}$  & The channel matrix of the AP-to-UE mmW path\\ \rowcolor[gray]{.9}
$\boldsymbol{H}_{g,t}$   & The channel matrix of the AP-to-reflector-to-UE path\\
$\tilde{\alpha}_{g,t}$ & The stochastic complex gain over path $g$ at time slot $t$\\ \rowcolor[gray]{.9}
$p_{0,t}^{(a)}$ & The beamforming-control policy of the mmW AP\\
$p_{g,t}^{(b)}$   & The phase shift-control policy of the mmW RIS\\ \rowcolor[gray]{.9}
$r_t$ & The achievable bitrate \\
$\mu$   & The risk sensitivity parameter\\ \rowcolor[gray]{.9}
$\mathcal{M}$ & The set of $M=G+1$ agents \\
$\mathcal{A}$   & The set of joint action space of the agents\\ \rowcolor[gray]{.9}
$\boldsymbol{s}_t$ & The state of POISG at time slot $t$ \\
$a_{m,t}$   & Action of agent $m$ at time slot $t$\\ \rowcolor[gray]{.9}
$T$ & Number of future consecutive time slots \\
$\Lambda_{T}$  & Trajectory of the POIPSG during $T$ time slots\\ \rowcolor[gray]{.9}
$R_{T,t}$ & rate summation during consecutive $T$ time slots \\
$\mathcal{H}_t$ & The global history at time slot $t$ \\ \rowcolor[gray]{.9}
$\mathcal{H}_{m,t}$   & The history for agent $m$ at time slot $t$\\
$\pi_{\boldsymbol{\theta}_m}$ & The parametric functional-form policy of agent $m$ \\ \rowcolor[gray]{.9}
$J(\boldsymbol{\theta},t)$  & The risk-sensitive episodic return at time slot $t$\\
$\Pi_{\boldsymbol{\theta}}(T)$   & Probability of trajectory during $T$ time slots under\\
& parametric policies$\{\pi_{\boldsymbol{\theta}_m}|\forall m \in \mathcal{M}\}$\\
\rowcolor[gray]{.9}
\hline
\end{tabular}
\end{center}
\label{tab:Symbols}
\vspace{-9mm}
\end{table}

\vspace{-0.4cm}
\subsection{Phase-shift controller for RIS-assisted mmW networks}\label{Sec:Optimum}
We define $\boldsymbol{P}_{0,t}=[p_{0,t'}^{(a)}]_{A \times T}$ as the beamforming-control policy of the mmW AP at time slot $t$, where $p_{0,t'}^{(a)}=\Pr(\theta_{0,t'}=\frac{-\pi}{2}+\frac{2a\pi}{A-1})$ is essentially the probability that the mmW AP selects the $a$-th beam angle from set $\Theta$ at time slot $t' \in \{t,...,t+T-1\}$. Next, we define $\boldsymbol{P}_{g,t}=[p_{g,t'}^{(b)}]_{B \times T}$ as the phase shift-control policy of the mmW RIS, where $p_{g,t'}^{(b)}=\Pr(\Psi_{gt'}=\Psi^{(b)})$ is the probability that the mmW RIS $g$ selects the $b$-th phase shift to reflect the received signal from the mmW AP toward the UE at time slot $t' \in \{t,...,t+T-1\}$.

Due to the stochastic changes of the mmW blockage between mmW AP or reflector and UE, and random changes in the user's orientation, the transmission and phase shift-control policies at a given slot $t$ will depend on unknown future changes in the LoS mmW links. Consequently, to guarantee ultra-reliable mmW links subject to future stochastic changes over mmW links, we mitigate the notion of risk instead of maximizing the expected future rate. Concretely, we adopt the entropic value-at-risk (EVaR) concept that is defined as $\frac{1}{\mu} \log\big( \mathbb{E}_{{r}_{t'}} \{e^{(-\mu \sum_{t'=t}^{t+T-1} {r}_{t'})}\} \big)$\cite{five}. Here, the operator $\mathbb{E}$ is the expectation operation. Expanding the Maclaurin series of the $\log$ and $\exp$ functions shows that EVaR takes into account higher order moments of the stochastic sum rate $\sum_{t'=t}^{t+T-1} {r}_{t'}$ during future $T$ consecutive time slots~\cite{four}. Consequently, we formulate the joint beamforming and phase shift-control problem for an RIS-assisted mmW network as follows:
\begin{align}
& \underset{\left\{
\substack{\boldsymbol{P}_{g,t},\\
\forall g \in \{0,...,G\}
}
\right\}}
\max \frac{1}{\mu} \log\big( \mathbb{E}_{{r}_{t'}} \{e^{(-\mu \sum_{t'=t}^{t+T-1} {r}_{t'})}\} \big), \label{Opt_prob}\\
& \hspace{0.1in} \sum_{a=1}^{A} p_{0,t'}^{(a)} = 1, \forall t' \in \{t,...,t+T-1\}, \label{Opt_prob_c1}\\
& \hspace{0.1in} \sum_{b=1}^{B} p_{g,t'}^{(b)} = 1, \forall t' \in \{t,...,t+T-1\},\forall g \in \{1,...,G\}, \label{Opt_prob_c2}\\
& \hspace{0.1in} 0\leq p_{0,t'}^{(a)} \leq 1, \forall a \in \{1,...,A\},\forall t' \in \{t,...,t+T-1\},\label{Opt_prob_c3}\\
& \hspace{0.1in} 0\leq p_{g,t'}^{(b)} \leq 1, \forall b \in \{1,...,B\},\forall t' \in \{t,...,t+T-1\},\nonumber \\ &
\forall g \in \{1,...,G\}, \label{Opt_prob_c4}
\end{align}
where the parameter $0 \leq \mu < 1$ denotes the risk sensitivity parameter~\cite{four}. In (\ref{Opt_prob}), the objective is to maximize the average of episodic sum of future bitrate, $\sum_{t'=t}^{t+T-1} {r}_{t'}$, while minimizing the variance to capture the rate distribution, using joint beamforming and phase shift-control policies of mmW AP and reflectors during future time slots. The risk sensitivity parameter penalizes the variance and skewness of the episodic sum of future bitrate. In (\ref{Opt_prob}),  $\{{r}_{t'}|t'=t,...,t+T-1\}$ depends on the beam angle of mmW AP, phase shift angle of mmW RIS, and the unknown AoA from user's location during future $T$-consecutive time slots.

The joint beamforming and phase shift-control problem in (\ref{Opt_prob}) is a stochastic optimization problem that does not admit a closed-form solution and has an exponential complexity~\cite{Sutton2018}. The complexity of the stochastic optimization problem in (\ref{Opt_prob}) becomes more significant due to the unknown probabilities for possible random network changes such as the mmW link blockages and the user's locations\cite{Sutton2018} as well as the large size of the state-action space. Moreover, since a mmW link can be blocked by the user's body, even if the user's location and surroundings around are fixed, we would still have self blockage by the user's body. Thus, even after initial beam tracking to find the location of these users, they can not be served in very long time. Therefore, we seek a low-complexity control policy to solve (\ref{Opt_prob}) that can intelligently  adapt to mmW link dynamics over future time slots. In this regard, we propose a framework based on principles of risk-sensitive deep RL and cooperative multi-agent system to solve the optimization problem in (\ref{Opt_prob}) with low complexity and in an adaptive manner.
\vspace{-0.8cm}
\section{Intelligent Beamfroming and Phase Shift-Control Policy}\label{Sec:Algorithm}
In this section, we present the proposed gradient-based and adaptive policy search method based on a new deep and risk-sensitive RL framework to solve the joint beamforming and phase shift-control problem in (\ref{Opt_prob}) in a coordinated and distributed manner. We model the problem in (\ref{Opt_prob}) as an identical payoff stochastic game (IPSG) in a cooperative multi-agent environment~\cite{Walidbook2000}. An IPSG describes the interaction of a set of agents in a Markovian environment in which agents receive the same payoffs\cite{Polecy2000}.

An IPSG is defined as a tuple $<\mathcal{S},\mathcal{M},\mathcal{A},\mathcal{O},T,R,o_0>$, where $\mathcal{S}$ is the state space, $\mathcal{M}=\{0,1,...,G\}$ is a set of $M=G+1$ agents in which index $0$ refers to the mmW AP and indexes 1 to $G$ represent the mmW RISs. $\mathcal{A}=\prod_{i \in \mathcal{M}} \mathcal{A}_i$ is the set of joint action space of the agents in which $\mathcal{A}_0=\Theta$ is the set of possible transmission directions for mmW AP and $\mathcal{A}_g=\Psi,\forall g=1,...,G$ is the set of possible phase shift for mmW RISs. The observation space $\mathcal{O}=\mathbb{R}$ is the bitrate over mmW link $r_{t}\in\mathcal{O}$. Here, $T:\mathcal{S} \times \mathcal{A}\rightarrow \Pr(\mathcal{S})$ is the stochastic state transition function from states of the environment, $s \in \mathcal{S}$ and joint actions of the agents, $\boldsymbol{a}\in \mathcal{A}$ to probability distributions over states of the environment, $T(\boldsymbol{s}',\boldsymbol{s},\boldsymbol{a}) = \Pr(\boldsymbol{s}_{t+1}= \boldsymbol{s}'|\boldsymbol{s}_t = \boldsymbol{s}, \boldsymbol{a}_t = \boldsymbol{a})$. $R(\boldsymbol{s}_t, \boldsymbol{a}_t)$ is the immediate reward function, and $o_0$ is the initial observation for the controllers of the mmW AP and reflectors\cite{three}.

Here, the immediate reward function, $R(\boldsymbol{s}_t, \boldsymbol{a}_t)$, is equal to the received bitrate which is given by (\ref{Bit_rate}). And the state $\boldsymbol{s}_t=\{\tilde{\boldsymbol{\alpha}}_{gt},\tilde{\boldsymbol{\beta}}_{gt}|g=0,...,G\}\cup\{\tilde{\gamma}_t\}\cup\{\tilde{\phi}_{t}\}$ includes complex path gains for all paths $g=0,...,G$ and AoA at UE at time slot $t$. Due to the dynamics over the mmW paths, the state, $\boldsymbol{s}_t$, and state transition function, $T(\boldsymbol{s}',\boldsymbol{s},\boldsymbol{a})$, are not given for the beamforming controller of mmW AP and phase shift-controllers of mmW RISs. Since all agents in $\mathcal{M}$ have not an observation function for all $s\in \mathcal{S}$, the game is a partially observable IPSG (POIPSG). Due to the partially observability of IPSG, a strategy for  agent $m$ is a mapping from the history of all observations from the beginning of the game into the current action $a_t$. Hereinafter, we limit our consideration to cases in which the agent has a finite internal memory including the history for agent $m$ at time slot $t$, $\mathcal{H}_{m,t}$. $\mathcal{H}_{m,t}=\{(a_{m,h},r_h)|h=t-H,...,t-1\}$ is a set of actions and observations for agent $m$ during $H$ consecutive previous time slots. We also define $\mathcal{H}_{t}=\cup_{m \in \mathcal{M}}\mathcal{H}_{m,t}$ as the global history.

Next we define a policy as the probability of action given past history as a continuous differentiable function of some set of parameters. Hence, we represent the policy of each agent $m$ of the proposed POIPSG in a parametric functional-form $\pi_{\boldsymbol{\theta}_m}(a_{m,t}|\mathcal{H}_{m,t})=\Pr\{a=a_{m,t}| \mathcal{H}_{m,t},\boldsymbol{\theta}\}$ where $\boldsymbol{\theta_m}$ is a parameter vector for agent $m$. If $\Lambda_{T}=\{(\boldsymbol{a}_{t'},r_{t'})|t'=t,...,t+T-1\}$ is a trajectory of the POIPSG during $T$-consecutive time slots, then the stochastic episodic reward function during future $T$-consecutive time slots is defined as $R_{T,t}=\sum_{t'=t}^{t+T-1}r_{t'}$. Here, we are interested in implementing a distributed controller in which the mmW AP and RISs act independently. Thus, the unknown probability of trajectory $\Lambda_{T}$ is equal to $\Pi_{\boldsymbol{\theta}}(T)=\prod_{t'=t}^{t+T-1} \prod_{m\in \mathcal{M}} \pi_{\boldsymbol{\theta}_m}(a_{m,t'}|\mathcal{H}_{m,t'}) \Pr\{r_{(t'+1)}|\boldsymbol{a}_{t'},\mathcal{H}_{m,t'}\}$ if the agents in $\mathcal{M}$ act independently.

In what follows we define the risk-sensitive episodic return for parametric functional-form policies $\pi_{\boldsymbol{\theta}_m},\forall m\in \mathcal{M},$ at time slot $t$ as $J(\cup_{m \in \mathcal{M}} \boldsymbol{\theta}_m ,t) =\frac{1}{\mu} \log\big( \mathbb{E}_{R_{T,t}} \{e^{(-\mu R_{T,t})}\} \big)$~\cite{four}. Given the parametric functional-form policies, $\pi_{\boldsymbol{\theta}_m},\forall m\in \mathcal{M}$, the goal of the transmission and phase shift controller is to solve the following optimization problem:
\begin{align}
& \underset{\left\{
\substack{ \cup_{m \in \mathcal{M}} \boldsymbol{\theta}_m }
\right\}}
\max J(\boldsymbol{\theta},t), \\
& \hspace{0.1in} 0\leq \pi_{\boldsymbol{\theta}_m}(a_{m,t'}|\mathcal{H}_{m,t'}) \leq 1, \forall a_{m,t'} \in \mathcal{A}_m,\forall m \in\mathcal{M},
\\ \nonumber
& \hspace{0.1in} \forall t' \in \{t,...,t+T-1\},\\
& \hspace{0.1in} \sum_{\forall a_{m,t'} \in \mathcal{A}_m} \pi_{\boldsymbol{\theta}_m}(a_{m,t'}|\mathcal{H}_{m,t'}) = 1, \forall m \in \mathcal{M},
\\ \nonumber
& \hspace{0.1in} \forall t' \in \{t,...,t+T-1\},\\
&  \hspace{0.1in} \boldsymbol{\theta}_m \in \mathbb{R}^{N}, \forall m \in \mathcal{M},
\label{Opt_prob_2}
\end{align}
where $T<<N$. We will define the parameter vector $\boldsymbol{\theta}$ and the value of $N$ in Subsection~\ref{Sec:Algorithm_A}.

To solve the optimization problem in (\ref{Opt_prob_2}), the controller needs to have full knowledge about the transition probability $\Pi_{\boldsymbol{\theta}}(T)$, and all possible values of $R_{T,t}$ for all of the trajectories during $t'=t,...,t+T-1$ from the POIPSG under policies $\pi_{\boldsymbol{\theta}_m},\forall m \in\mathcal{M}$. Since the explicit characterization of the transition probability and values of the episodic reward for all of the trajectories is not feasible in highly dynamic mmW neworks, we use an RL framework to solve (\ref{Opt_prob_2}). More specifically, we use a \emph{policy search approach} to find the optimal transmission angle and phase shift-control policies to solve problem in (\ref{Opt_prob_2}) for the following reasons. First, value-based approaches such as $Q$-learning are oriented toward finding deterministic policies. However, the optimal policy is often stochastic and policy-search approaches can select different phase shifts with specific probabilities by adaptively tuning the parameters in $\boldsymbol{\theta}$ ~\cite{Sutton2018}. Second, value-based RL methods are oriented toward finding deterministic policies, and they use a parameter, $\epsilon$, as an exploration-exploitation tradeoff to apply other possible policies~\cite{Sutton2018}. However, In policy search approach, the exploration-exploitation tradeoff is explicitly applied due to the direct modeling of probabilistic policy~\cite{Sutton2018}. Third, any small change in the estimated value of an action can cause it to be (or not) selected in the value-based approaches. In this regard, the most popular policy-search method is the policy-gradient method where the gradient objective function is calculated and used in gradient-ascend algorithm. The gradient $\nabla J(\boldsymbol{\theta},t)$ of the risk-sensitive objective function is approximated as follows.
\begin{proposition}
\textnormal{The gradient of the objective function, $J(\boldsymbol{\theta},t)$, in (\ref{Opt_prob_2}) is approximated by:}
\begin{align}
&\nabla_{\boldsymbol{\theta}}J(\boldsymbol{\theta},t)\approx
 \mathbb{E}_{\Lambda_{T}}\{\nabla_{\boldsymbol{\theta}}\log\Pi_{\boldsymbol{\theta}}(T) \times
\nonumber \\ &
\big((1+\mu \mathbb{E}_{\Lambda_{T}}\{R_{T,t}\}) R_{T,t}-\frac{\mu}{2}R_{T,t}^2\big)\},
\label{gradient_equation}
\end{align}
\textnormal{where $\mathbb{E}_{\Lambda_{T}}\{R_{T,t}\}=\sum_{\Lambda_{T}}\Pi_{\boldsymbol{\theta}}(T) R_{T,t}$. Under distributed controller in which mmW AP and RISs act independently, $ \nabla_{\boldsymbol{\theta}}\log\Pi_{\boldsymbol{\theta}}(T)=\sum_{t'=t}^{t+T-1} \sum_{m \in \mathcal{M}} \nabla_{\boldsymbol{\theta}_m} \log\pi_{\boldsymbol{\theta}_m}(a_{m,t'}|\mathcal{H}_{m,t'})$.
}
\end{proposition}
\begin{proof}
See Appendix A.
\end{proof}

Following Proposition 1, we can use (\ref{gradient_equation}) to solve the optimization problem in (\ref{Opt_prob_2}) using a gradient ascent algorithm and, then, find the near optimal control policies. To calculate (\ref{gradient_equation}), we need a lookup table of all trajectories of risk-sensitive values and policies over time. However, this lookup table is not available for a highly dynamic indoor mmW networks. To overcome this challenge, we combine DNN with the RL policy-search method. Such a combination was studied in~\cite{three}, where a DNN learns a mapping from the partially observed state to an action without requiring any lookup table of all trajectories of the risk-sensitive values and policies over time. Next, we propose an RL algorithm that uses a DNN based on policy gradient for solving (\ref{Opt_prob_2}).
\vspace{-0.4cm}
\subsection{Implementation of Phase-shift controller with DNN}\label{Sec:Algorithm_A}
We use a DNN to approximate the policy $\pi_{\boldsymbol{\theta}_{m}},\forall m \in \mathcal{M}$ for solving (\ref{Opt_prob_2}). Here, the parameters $\boldsymbol{\theta} \in \mathbb{R}^{N}$ include the weights over all connections of the proposed DNN where $N$ is equal to the number of connections~\cite{three}. We consider two implementations of the beamforming and phase shift-controllers: centralized and distributed.

1) \emph{Centralized controller:}
the centralized controller has enough memory to record the global history $\mathcal{H}_t=\cup_{m \in \mathcal{M}} \mathcal{H}_{m,t}$ and computational power to train the proposed RNN in Fig.~\ref{System_model_21}. Thus, the deep RNN directly implements the independent beamforming and phase shift-control policies $\pi_{\boldsymbol{\theta}}(\boldsymbol{a}_{t'}|\mathcal{H}_{t'})$ for $t'=t,...,t+T-1$ given the global history $\mathcal{H}_t$ and $\boldsymbol{\theta}=\cup_{m \in \mathcal{M}}\boldsymbol{\theta}_m$. Then, the policy is transmitted from the centralized controller to the mmW AP and RISs through the control links. Indeed, the centralized controller is a policy mapping observations to the complete joint distribution over set of joint action space $\mathcal{A}$. The deep RNN that implements the centralized controller is shown in Fig.~\ref{System_model_21}. This deep RNN includes 3 long short term memory (LSTM) cells, 3 fully connected, 3 rectified linear unit (Relu), and $M$ Softmax layers. The 3 LSTM layers have layers of $H$, $\frac{H}{2}$, and $\frac{H}{4}$ memory cells.

The main reason for using the RNN to implement the controller is that unlike feedforward neural networks (NNs), the RNNs can use their internal state to process sequences of inputs. This allows RNNs to capture the dynamic temporal behaviors of a system such as highly dynamic changes over mmW links between mmW AP and reflectors in an indoor scenarios~\cite{Sutton2018}. Thus, we implement the controller using LSTM networks. An LSTM is an artificial RNN architecture used in the field of deep learning. In this case, the LSTM-based controller has enough memory cell in LSTM layers to learn policy that require memories of events over previous discrete time slots. These events are the blockage of mmW links due to the stochastic state transition function from states of the environment in the proposed POIPSG during last time slots. Moreover, the LSTM-based architecture allows us to avoid the problem of vanishing gradients in the training phase. Hence, LSTM-based architecture and compared to other DNNs provides a faster RL algorithm~\cite{Sutton2018}.
\begin{figure}[!t]
	\begin{center}
		\includegraphics[width=8.5cm]{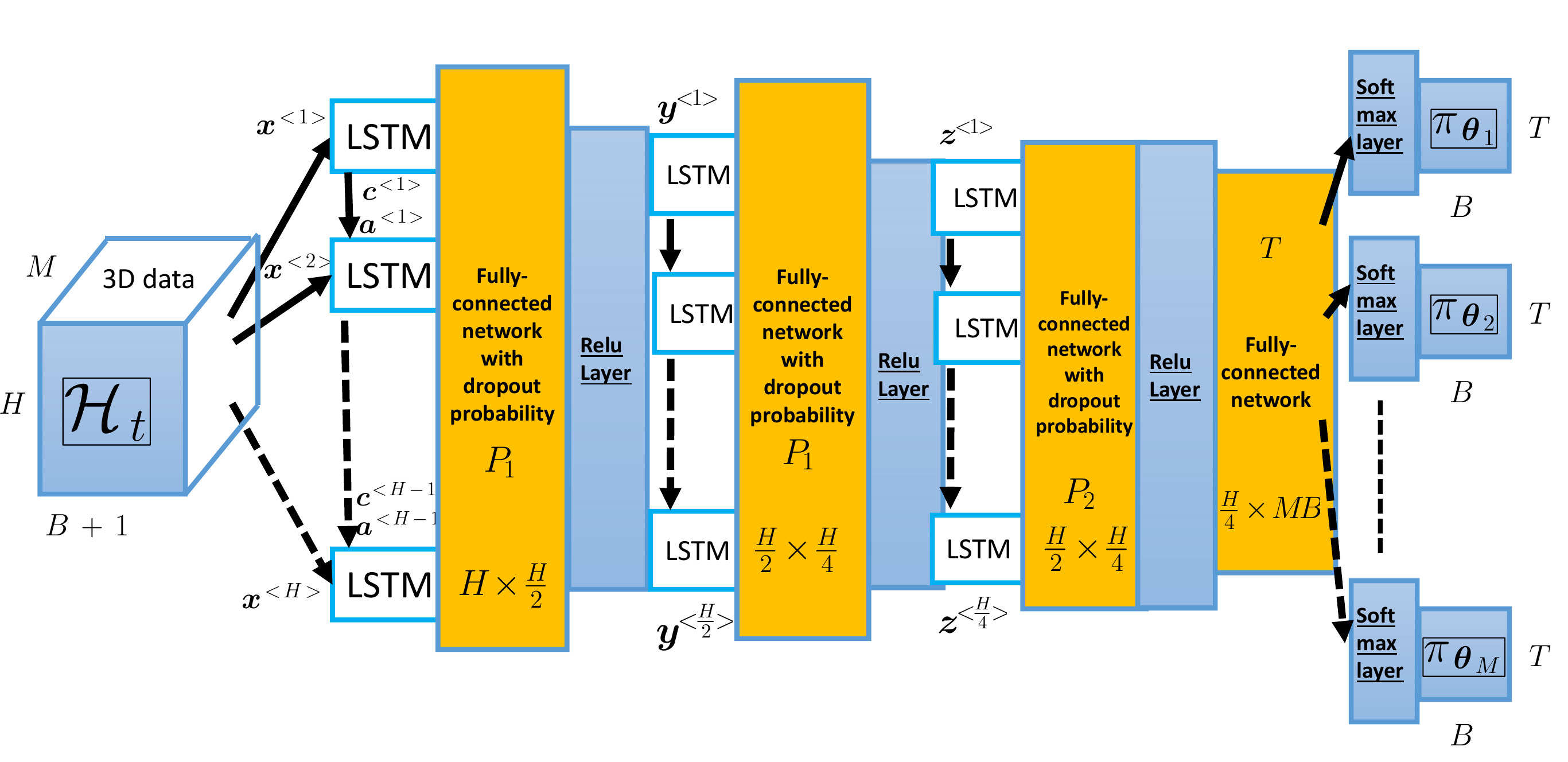}		\vspace{-0.2cm}
		\caption{ \small The deep RNN for implementing the centralized controller. Input is $\mathcal{H}_{t}$ and output is $\cup_{m \in \mathcal{M}}\pi_{\boldsymbol{\theta_m}}$. \vspace{-1cm} }
		\label{System_model_21}
	\end{center}
\end{figure}

2) \emph{Distributed controllers:}
in the highly dynamic mmW network, even during the policy signal transmissions over backhaul link from central controller to the mmW AP and RISs, the channel state may change. So, we have proposed a distributed control policy in which each of mmW AP or RISs will optimized their control policy in a distributed manner without requiring to send central policy over backhaul link. In the distributed controllers, the mmW AP and all the RISs act independently. In this case, since each agent acts independently, $\Pi_{\boldsymbol{\theta}}(T)=\prod_{m\in \mathcal{M}} \pi_{\boldsymbol{\theta}_m}$, and each deep RNN, which is in the controller of each agent $m$, implements the policy $\pi_{\boldsymbol{\theta}_m}$ because of the limited computational power. Although the mmW AP and RISs act independently, agents share their previous $H$ consecutive actions with other agents  using the synchronized coordinating links between themselves. The deep RNN that implements the distributed controller of each agent $m$ is shown in Fig.~\ref{System_model_22}. This deep RNN includes 2 LSTM, 3 fully connected, 2 Relu, and one Softmax layer. The two LSTM layers have layers of $H$ and $\frac{H}{4}$ memory cells.
\begin{figure}[!t]
	\begin{center}
		\includegraphics[width=8.5cm]{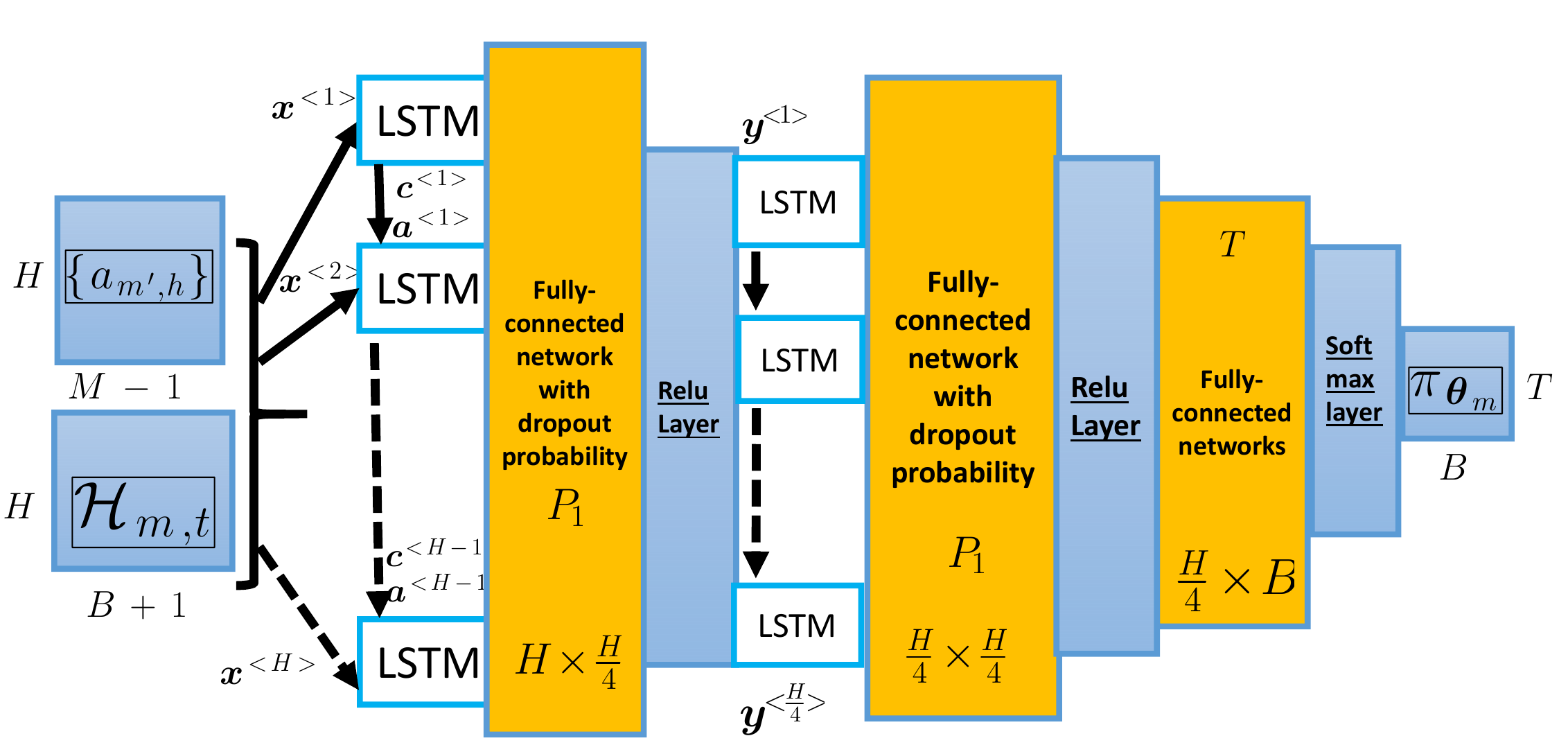}		\vspace{-0.4cm}
		\caption{ \small The deep RNN for implementing the distributed phase shift-controller. Input is $\mathcal{H}_{m,t}$ and output is $\pi_{\boldsymbol{\theta_m}}$.
\vspace{-1cm}}
		\label{System_model_22}
	\end{center}
\end{figure}

One of the techniques that can be used to prevent an NN from overfitting the training data is the so-called dropout technique~\cite{Goodfellow-2016}. We will find the value for dropout probabilities $P_1$ and $P_2$ for our proposed deep NN in Figs.~\ref{System_model_21} and~\ref{System_model_22} using trial-and-error procedure in the simulation Section. Since the payoff is identical for all agents and the observation of environment changes is from the same distribution for all agents, the gradient updating rules of the distributed and central controllers will be same in the considered POIPSG. This fact is shown as follow:
\begin{theorem}
\textnormal{Starting from the same point in the search space of policies for the proposed POIPSG and given the  identical payoff function, $J(\boldsymbol{\theta},t)$, the gradient update algorithm will converge to the same locally optimal parameter setting for the distributed controllers and centralized controller.}
\end{theorem}
\begin{proof}
See Appendix B.
\end{proof}

Following Theorem 1, if the architectures of the centralized controller in Fig.~\ref{System_model_21}, and distributed controllers in Fig.~\ref{System_model_22} are designed correctly and the proposed deep RNNs are trained with enough data, the performance of distributed controllers should approach that of the centralized controller in the RIS-assisted mmW networks. In this case, instead of using a central server in the RIS-assisted mmW networks with highly computational cost and signaling overhead to send the control policies to all agents across all network, one can use the distributed coordinated controllers with low computational power. Moreover, for an indoor environment with a large number of dark areas, more RISs are required, basically one RIS per dark area. Thus, compare to centralized solution, a distributed control policy is required to guarantee scalability of our proposed solution for the environment with high number of RISs. In this case, the distributed controllers just need to share the policies with the agents that cooperate to cover the same dark area. Thus, the signaling overhead is also limited to the local area in the distributed controller setting. In addition to these, the policy profile under the distributed controllers is a Nash equilibrium of the POIPSG. We state this more precisely in the following.
\begin{theorem}
\textnormal{At the convergence of the gradient update algorithm in (\ref{gradient_equation}), the policy profile under the distributed controllers is an NE equilibrium of the POIPSG.}
\end{theorem}
\begin{proof}
See Appendix C.
\end{proof}

Consider a training set $\mathcal{S}$ of $S$ samples that is available to train the deep RNN network. Each training sample $s$ includes a sample of policies and bitrates during $H$-consecutive time slots before time slot $t_s$, $\{\pi_{\boldsymbol{\theta}_m}^{(s)}(a_{t'}|r_{t'}),r_{t'}^{(s)}|t'=t_s-h+1,...,t_s, \forall m\in \mathcal{M} \}$, and policies and bitrates during future $T$-consecutive time slots after time slot $t_s$,
$\{\pi_{\boldsymbol{\theta}_m}^{(s)}(a_{t'}|r_{t'}),a_{t'}^{(s)},r_{t'}^{(s)}|t'=t_s+1,...,t_{s}+T, \forall m\in \mathcal{M}\}$. Consequently, based on Proposition 1 and by replacing the expectation with sample-based estimator for $\nabla_{\boldsymbol{\theta}}J(\boldsymbol{\theta})$, we use the gradient-ascend algorithm to train the RNN  as follows:
\begin{align}
& \nabla_{\boldsymbol{\theta}}J(\boldsymbol{\theta})\approx \frac{1}{S}\sum_{s=1}^S \Big(\nabla_{\boldsymbol{\theta}}\log\Pi_{\boldsymbol{\theta}}^{(s)}(T) \times \big((1-\mu R_S) R_{T,t_s}+\nonumber \\ &\frac{\mu}{2}R_{T,t_s}^2\big)\Big),
\boldsymbol{\theta}\leftarrow \boldsymbol{\theta}+\alpha \nabla_{\boldsymbol{\theta}}J(\boldsymbol{\theta}), \label{Learning-algorithm}
\end{align}
where $R_{T,t_s}=\sum_{t'=t_s+1}^{t_s+T} r_{t'}^{(s)}$, and $R_S=\frac{1}{S}\sum_{s=1}^S R_{T,t_s}$. Here, $\alpha$ is the learning rate. In summary, to solve the optimization problem in (\ref{Opt_prob}), we model the problem using deep and risk-sensitive RL framework as the problem (\ref{Opt_prob_2}). Then, to solve the problem (\ref{Opt_prob_2}), we implement two centralized and distributed policies using deep RNNs which are shown Figs.~\ref{System_model_21} and~\ref{System_model_22}. Then, based on gradient ascent algorithm, we use (\ref{Learning-algorithm}) to iteratively train the proposed deep RNNs and optimize $\boldsymbol{\theta_m},\forall m \in \mathcal{M}$. Algorithm~\ref{algorithm-E}  presents the deep RNN-based RL approach of our proposed joint mmW beamforming and RIS phase shift changing control policy. We should note that in addition to policy gradient approach, other on policy RL learning algorithms such as proximal policy optimization (PPO) can be applied in our proposed framework. Indeed, PPO approach will lead to find the stable policy with lower variance in the process of policy search. However, compare to policy gradient approach, the PPO also need more iteration to achieve convergence. In practice, the proposed deep RNNs in Figs.~\ref{System_model_21} and~\ref{System_model_22} can be run directly on the FPGA fabric of a software-defined radio (SDR) platform such as $\text{DeepRadio}^{\text{TM}}$~\cite{9082619}.
\vspace{-0.2cm}

\begin{algorithm}[!t]
\caption{\small{Intelligent beamforming and phase-shift control policy}} \label{algorithm-E}
\small{
\begin{algorithmic}[1]
\State \textbf{Input:} Set of mmW AP and RISs: $\mathcal{M}$; initial training set $\mathcal{S}=\{\mathcal{H}^{(s)}_{m,t_s},r_{t'}^{(s)}|t'=t_s+1,...,t_s+T, \forall m\in \mathcal{M}\}$ of $S$ samples of histories and bitrates;
\State \textbf{Phase I - Network Operator}
\State Tune risk sensitivity parameter, $\mu$, to  maximize the expected bitrate and mitigate the risk of mmW link blockage;
\State Define a deep RNN-based control policy mode,$\{\pi_{\boldsymbol{\theta}_m}|,\forall m\in \mathcal{M}\}$, e.g., distributed or centralized, shown in Figs. 2 or 3 of the revised manuscript, respectively.
\State \textbf{Phase II - Offline training}
\State Train the deep RNN-based control policy, $\{\pi_{\boldsymbol{\theta}_m}|,\forall m\in \mathcal{M}\}$ ,using initial training set following the gradient-ascend algorithm in (14) of the revised manuscript, $\boldsymbol{\theta}\leftarrow \boldsymbol{\theta}+\alpha \nabla_{\boldsymbol{\theta}}J(\boldsymbol{\theta})$ , where $\nabla_{\boldsymbol{\theta}}J(\boldsymbol{\theta})\approx \frac{1}{S}\sum_{s=1}^S \Big(\nabla_{\boldsymbol{\theta}}\log\Pi_{\boldsymbol{\theta}}^{(s)}(T) \times \big((1-\mu R_S) R_{T,t_s}+\frac{\mu}{2}R_{T,t_s}^2\big)\Big)$,
\State \textbf{Phase III - Online deep reinforcement learning}
\Repeat
\parState {Observe the global history, $\mathcal{H}_t$, at time slot $t$;}
\For{each agent $m\in \mathcal{M}$}
\parState{if $m=0$, mmW AP forms transmission beam using $\pi_{\boldsymbol{\theta}_0}$ policy;}
\parState{if $m\neq0$, RIS $m$ shifts the phase of received signals using $\pi_{\boldsymbol{\theta}_m}$ policy;}
\EndFor
\State during future $T$-consecutive time slots, perform control policies, $\{\pi_{\boldsymbol{\theta}_m}|,\forall m\in \mathcal{M}\}$, and capture the received bitrate;
\State Update the training set $\mathcal{S}=\mathcal{S} \cup \{\mathcal{H}_{m,t},r_{t'}|t'=t+1,...,t+T, \forall m\in \mathcal{M}\}$;
\State Uniformly select a set of $S_{b}$ samples from updated training set $S$ as a minibatch set
\State Update he deep RNN-based control policy following the gradient-ascend algorithm in (14) of the revised manuscript, $\boldsymbol{\theta}\leftarrow \boldsymbol{\theta}+\alpha \nabla_{\boldsymbol{\theta}}J(\boldsymbol{\theta})$ , where $\nabla_{\boldsymbol{\theta}}J(\boldsymbol{\theta})\approx \frac{1}{S}\sum_{s=1}^S \Big(\nabla_{\boldsymbol{\theta}}\log\Pi_{\boldsymbol{\theta}}^{(s)}(T) \times \big((1-\mu R_S) R_{T,t_s}+\frac{\mu}{2}R_{T,t_s}^2\big)\Big)$,
\State $t=t+T$;
\Until {$t=t_{\text{end}}$ or convergence happens}
\State \textbf{Phase IV - Stable control policy}
\State \textbf{Ouput:} the stable beamforming and phase-shift control policy profile, $\{\pi_{\boldsymbol{\theta}^*_m}|,\forall m\in \mathcal{M}\}$, that is a Nash equilibrium of the POIPSG under distributed controllers or sub-optimal solution for problem (5) in the revised manuscript under centralized controllers.
\end{algorithmic}
}
\end{algorithm}

\subsection{Complexity of deep RNN-based policies}
The complexity of an NN depends on the number of hidden layers, training examples, features, and nodes at each layer~\cite{complexity}. The complexity for training a neural network that has $L$ layers and $n_l$ node in layer $l$ is given by $\mathcal{O}(nt\prod_{l=1}^{L-1} n_ln_{(l+1)})$ with $t$ training examples and $n$ epoch. Meanwhile, the complexity for one feedforward propagation will be $\mathcal{O}(\prod_{l=1}^{L-1} n_ln_{(l+1)})$. On the other hand, LSTM is local in space and time, which means that the input length does not affect the storage requirements of the network~\cite{complexity3}. In practice, after training the RNN-based policy, our proposed solution will use the feed-forward propagation algorithm to find the solution. In this case, following the proposed RNN architectures in Figs.~\ref{System_model_21} and~\ref{System_model_22}, the complexities of the centralized and distributed controllers are $\mathcal{O}(H(H+MBT))$ and $\mathcal{O}(H(H+MB+BT))$, respectively. These complexities are polynomial functions of key parameters such as history length, $H$, number of mmW AP and RISs, $M$, phase shift angles, $B$, and future time slots, $T$. On the other hand the complexity of optimal solution suing brute force algorithm is $\mathcal{O}(MBT^2+MBTH)$. Consequently, for a given history length $H$, the optimal solution has the highest complexity, $\mathcal{O}(MBT^2)$, while the complexity of our proposed distributed solution, $\mathcal{O}((MB+BT))$, is the least complex.

\section{Simulation Results and Analysis}\label{Sec:Simulation}
For our simulations, the carrier frequency is set to 73 GHz and the mmW bandwidth is 1 GHz. In this case, the value of the wavelength lambda of the carrier frequency is $\lambda=\frac{c}{f}=\frac{3\times 10^8}{73 \times 10^9}\simeq 4$ mm. The number of transmit antennas at the mmW AP and receive antennas at the UE are set to 128 and 64, respectively. The duration of each time slot is 1 millisecond which is consistent the mmW channel coherence time in typical indoor environments~\cite{sevenn}. The transmission power of the mmW AP is 46 dBm and the power density of noise is -88 dBm.  We assume that the mmW RIS assigns a square of $8\times8=64$ meta-surfaces to reflect the mmW signals. Each meta-surface shifts the phase of the mmW signals with a step of $\frac{\pi}{5}$ radians from the range $[-\frac{\pi}{2},\frac{\pi}{2}]$. In our simulation, we assume that one mmW AP and two mmW RISs are mounted on the walls of the room and controlled using our proposed framework to guarantee reliable transmission. To evaluated our proposed RNN-based control policies, we use two real-world and model-based datasets of the users' trajectories in an indoor environment. To generate model-based dataset, we consider a 35-sq. meter office environment with a static wall blockage at the center. In this regard, we have assumed a given probability distribution for the users' location in a room.  This location probability distribution can be calculated using well-known indoor localization techniques such as the one in~\cite{three}. For generating the data set of mobile users' trajectories, we use a modified random walk model. In this case, the direction of each user's next step is chosen based on the probability of user's presence at next step location. Fig.~\ref{Simulation_model} shows the probability distribution of the user's locations in the office, the location of the mmW RIS, and an illustrative example of a user trajectory. We further evaluate our proposed solution using real-world dataset. We use the OMNI1 dataset~\cite{dataset}. This dataset includes trajectories of humans walking through a lab captured using an omni-directional camera. Natural trajectories collected over 24 hours on a single Saturday. This dataset contains 1600 trajectories during 56 time slots. For comparison purposes, we consider the optimal solution, as a benchmark in which the exact user's locations and optimal strategies for the reflector during the next future $T$-time slots are known.
\vspace{-0.1cm}
\begin{figure}[!t]
	\begin{center}
		\includegraphics[width=8cm]{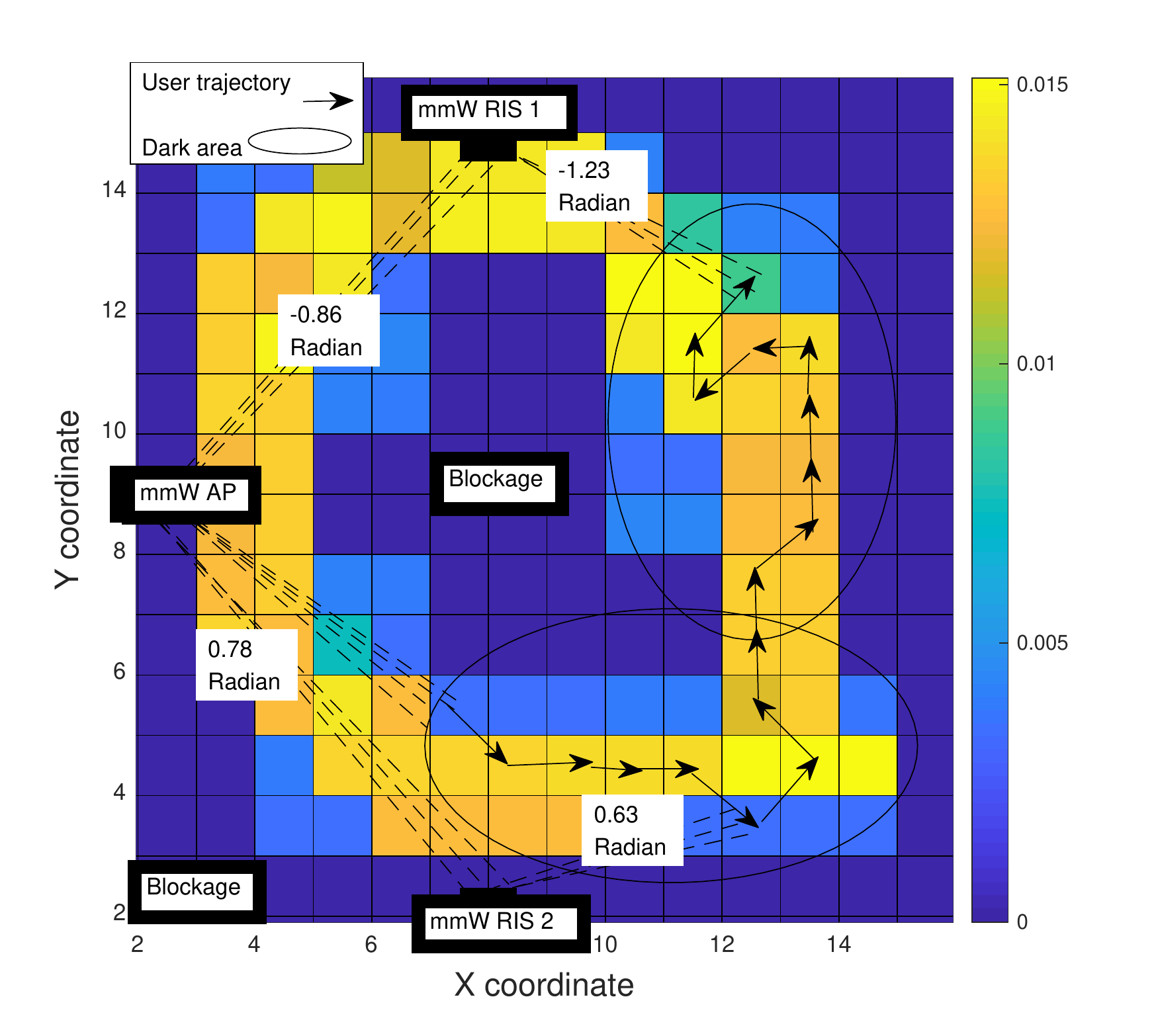}	
	    \vspace{-0.1cm}
		\caption{ \small The distribution probability of mobile user's location.}
        \vspace{-1cm}
		\label{Simulation_model}
	\end{center}
\end{figure}

\subsection{Performance evaluation of deep RNN training}
To evaluate the performance of the proposed controllers implemented with deep RNN depicted in Figs.~\ref{System_model_21} and~\ref{System_model_22}, Fig.~\ref{RMSE} shows the RMSE between the predicted and optimal policies of the centralized and distributed controllers when dropout probabilities are $P_{1}=0.2$ and $P_{2}=0.4$. On average the difference between RMSEs over the training and validation sets is less than $1\%$ which shows that the deep RNN model is not over-fitted to the training data set. In addition, on average the difference between RMSEs over training and test sets is less than $0.7\%$ which shows that implemented deep RNN model is not under-fitted and the deep RNN model can adequately capture the underlying structure of the new dynamic changes over mmW links. Thus, the structure of proposed deep RNN models depicted in~\ref{System_model_21} and~\ref{System_model_22} are correctly chosen and the hyper-parameters such as dropout probabilities $P_{1}=0.2$ and $P_{2}=0.4$ in the training phase are tuned correctly. On average, the RMSE for future consecutive time slots is $5.5\%$ for $T=2$ and $11.5\%$ for $T=4$. This show that predicting the correct control strategy becomes harder when the window length of future consecutive time slots increases, but even for $T=4$ the deep RNN can capture the unknown future dynamics over mmW links and correctly predict control strategy in $88.5\%$ of times. Beside these, the differences of RMSEs between centralized and distributed controllers are $0.3\%$, $1\%$, and $0.9\%$ over training, validation, and test sets. This shows that the performance of centralized and distributed controllers are almost as same as each others.
\begin{figure}[t!]
  \centering
  \includegraphics[width=7cm]{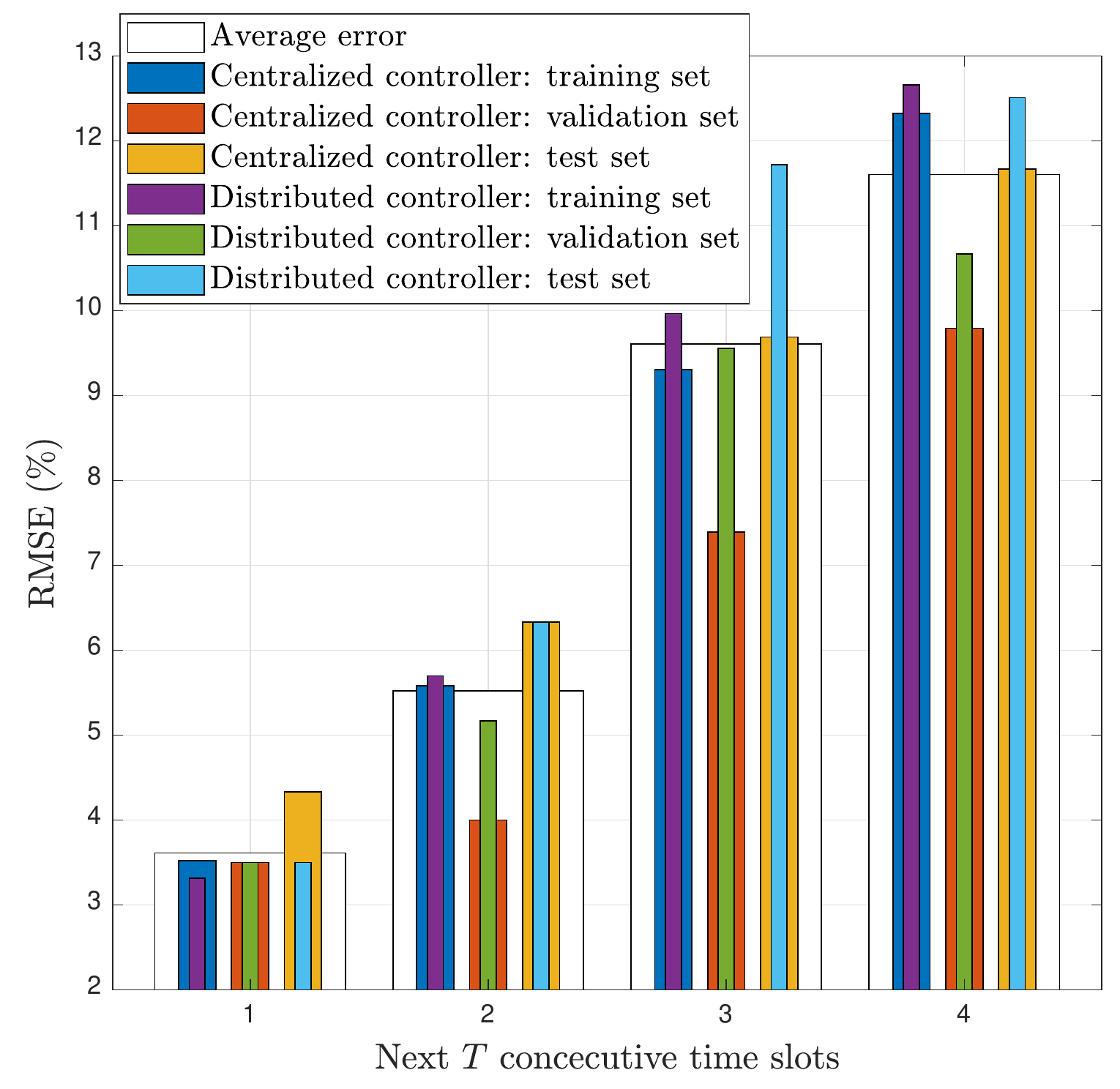}
  \vspace{-0.1cm}
  \caption{RMSE for the parametric functional-form policy.}
  \vspace{-0.5cm}
  \label{RMSE}
\end{figure}

\subsection{Achievable rate under proposed RNN-based controllers}
In Fig.~\ref{Rate_Time}, we show the achievable rate, $R_T$, following the centralized and distributed controller policies over time for model-based dataset presented in the simulation setup. As we can see from Figs.~\ref{Rate_Time_Dist} and~\ref{Rate_Time_Cent}, when the risk sensitivity parameter is set to zero, called i.e., non-risk scenario, a higher rate with highly dynamic changes is achieved under the optimal solution. However, when the risk sensitivity parameter increases from $0$ to $0.8$, i.e., risk-based scenario, the policy resulting from the centralized and distributed controllers achieves less average rate with lower variance which is more reliable. For model-based datset, on average, the mean and variance of the achievable rate for the non-risk scenario are $28\%$ and $60\%$ higher than the risk-based scenarios for different future time slot lengths, respectively. Moreover, we can also see that, controlling during wider time window of future consecutive time slots leads to more reliable achievable rate but with lower average rate for the risk-based scenario. For example, when $T=2$, the mean and variance of the achievable rate are $7.27$ and $0.053$ respectively, but the mean and variance of achievable rate respectively decrease to $3.92$ and $0.0018$ when $T=4$. The reason is that controlling the beam angle of mmW AP and phase shift of RISs for larger window of future time slots gives the centralized and distributed controllers more available strategy to decrease the variance more compare to controlling the beam angle and phase shift during tighter window of future time slots. In addition to this, on average, the mean of the rate achieved by the distributed controller is $4.5\%$ higher than the centralized controller and the difference in the variance of the achieved rate between the centralized and distributed controllers is $2\%$. This result shows that the performance of the centralized and distributed controllers is identical.
\begin{figure}[ht]
\begin{subfigure}{.5\textwidth}
  \centering
  \includegraphics[width=7cm]{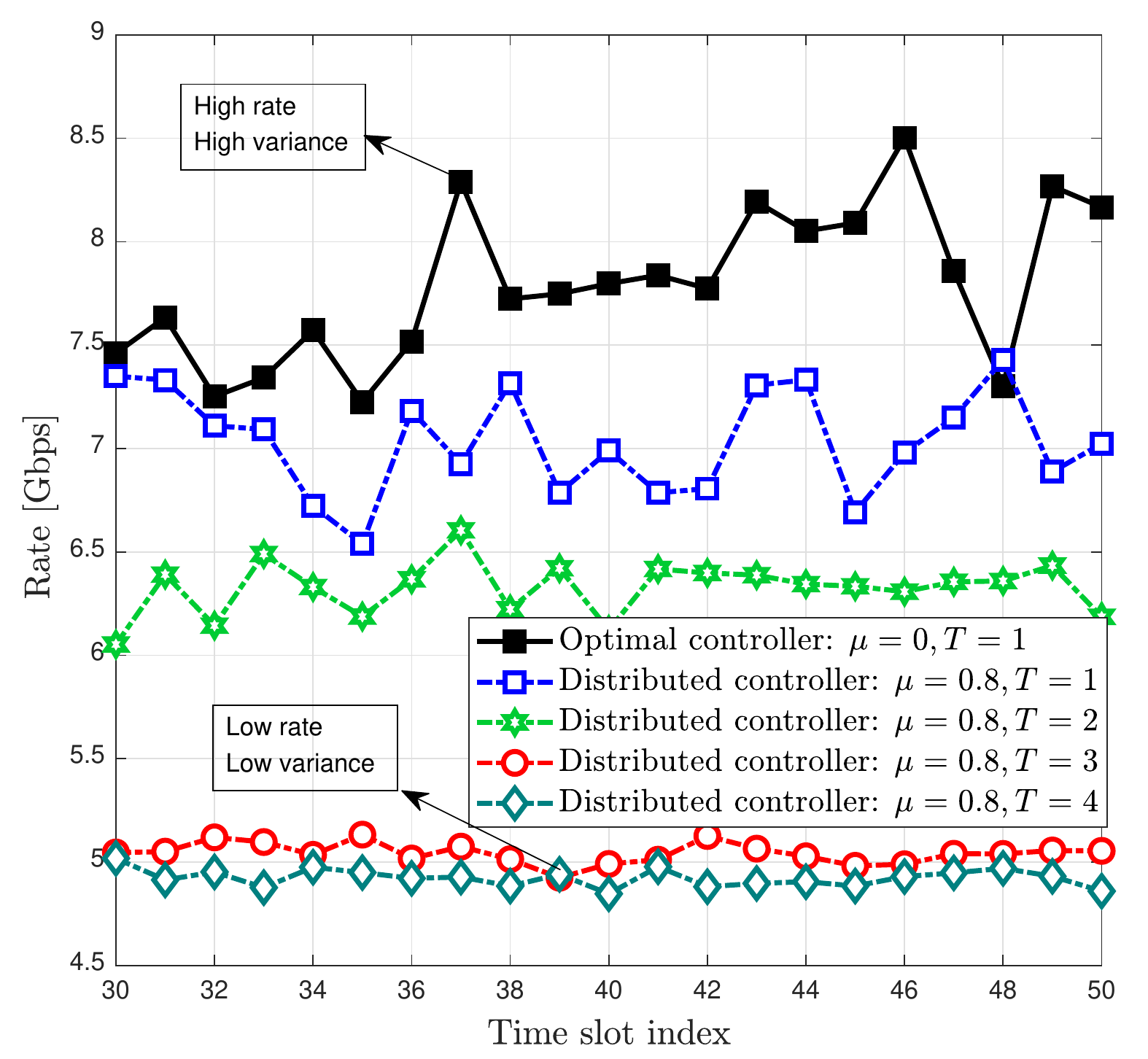}
  \caption{Distributed controllers.}
  \label{Rate_Time_Dist}
\end{subfigure}
\begin{subfigure}{.5\textwidth}
  \centering
 \includegraphics[width=7cm]{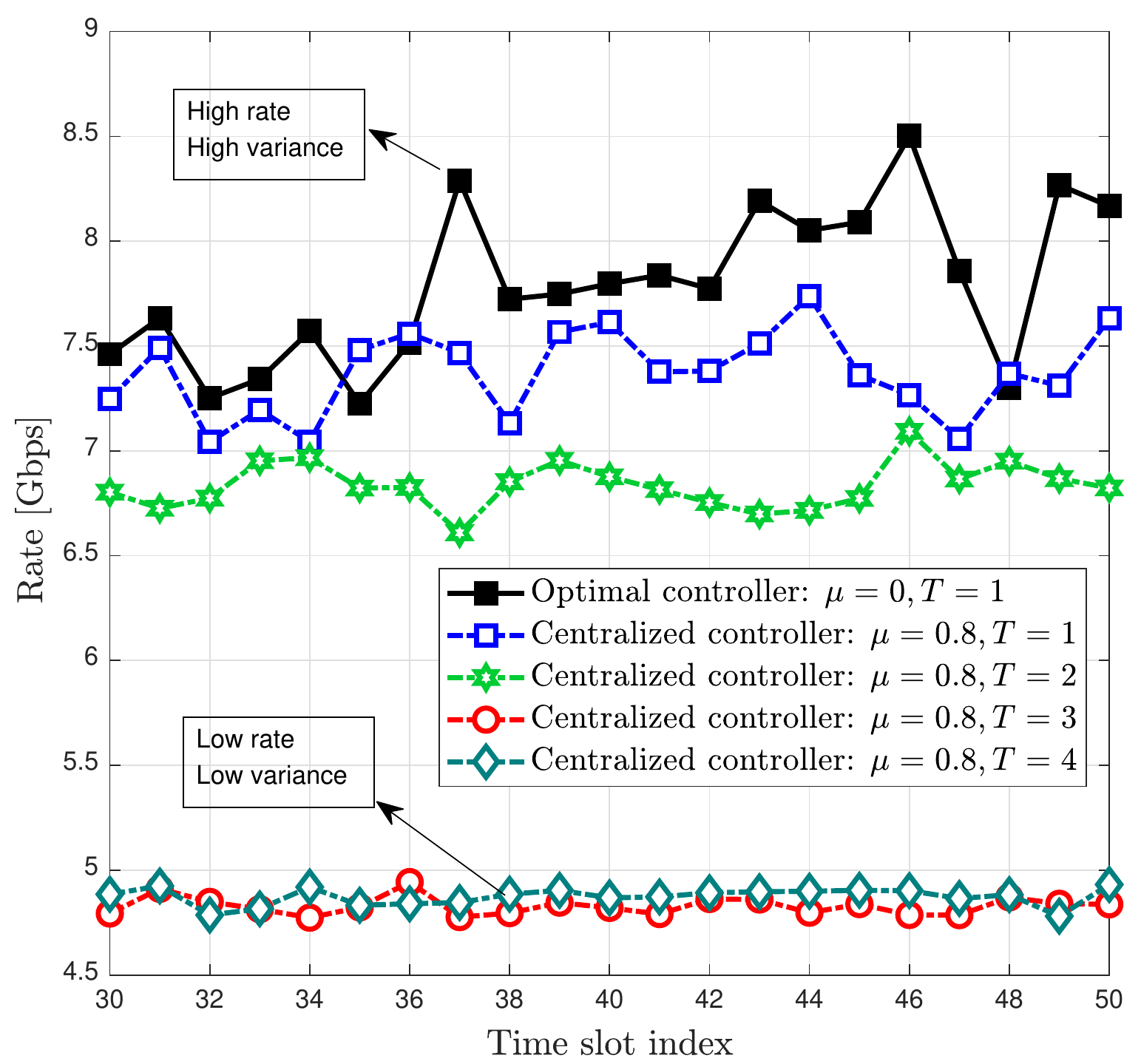}
  \caption{Centralized controller.}
  \label{Rate_Time_Cent}
\end{subfigure}
\caption{Achievable rate, $R_T$, for model-based dataset.}\label{Rate_Time}
\end{figure}

In Fig.~\ref{Rate_Time_Lab}, we show the achievable rate, $R_T$, following the centralized and distributed controller policies for real-world dataset in~\cite{dataset}. From Figs.~\ref{Rate_Time_Lab_scenario_Dist} and~\ref{Rate_TimeLab_scenario_Central}, we observe that, in a non-risk scenario, $\mu=0$, a high rate with high variance is achieved under the optimal solution. However, in a risk-based scenario, $\mu=0.8$, the policy resulting from the centralized and distributed controllers achieves a smaller data rate but with lower variance which is more reliable. For real-world dataset, on average, the mean and variance of the achievable rate for the non-risk scenario are $17\%$ and $34\%$ higher than the risk-based scenarios for different future time slot lengths, respectively. Moreover, when $T=2$, the mean and variance of the achievable rate are $4.45$ and $0.0066$ respectively, but the mean and variance of the achievable rate respectively decrease to $3.31$ and $0.0023$ when $T=4$. In addition to this, on average, the differences in the variance and the mean of the rate achieved by the centralized and distributed controllers  are $6\%$ and $0.8\%$. This result shows that the performance of the centralized controller is near the distributed controller performance for real-world dataset.

\begin{figure}[ht]
\begin{subfigure}{.5\textwidth}
  \centering
  \includegraphics[width=7cm]{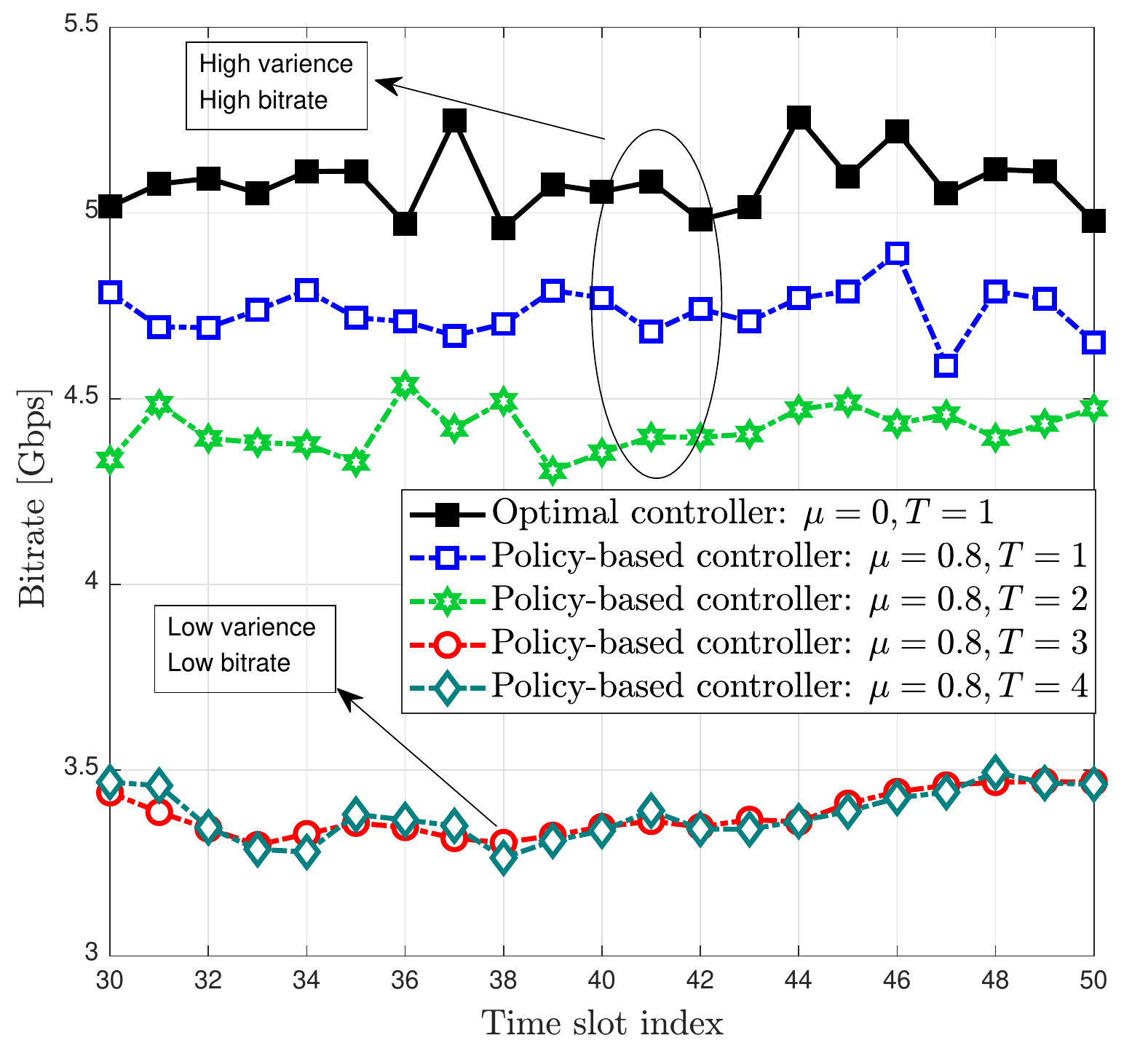}
  \caption{Distributed controllers.}
  \label{Rate_Time_Lab_scenario_Dist}
\end{subfigure}
\begin{subfigure}{.5\textwidth}
  \centering
 \includegraphics[width=7cm]{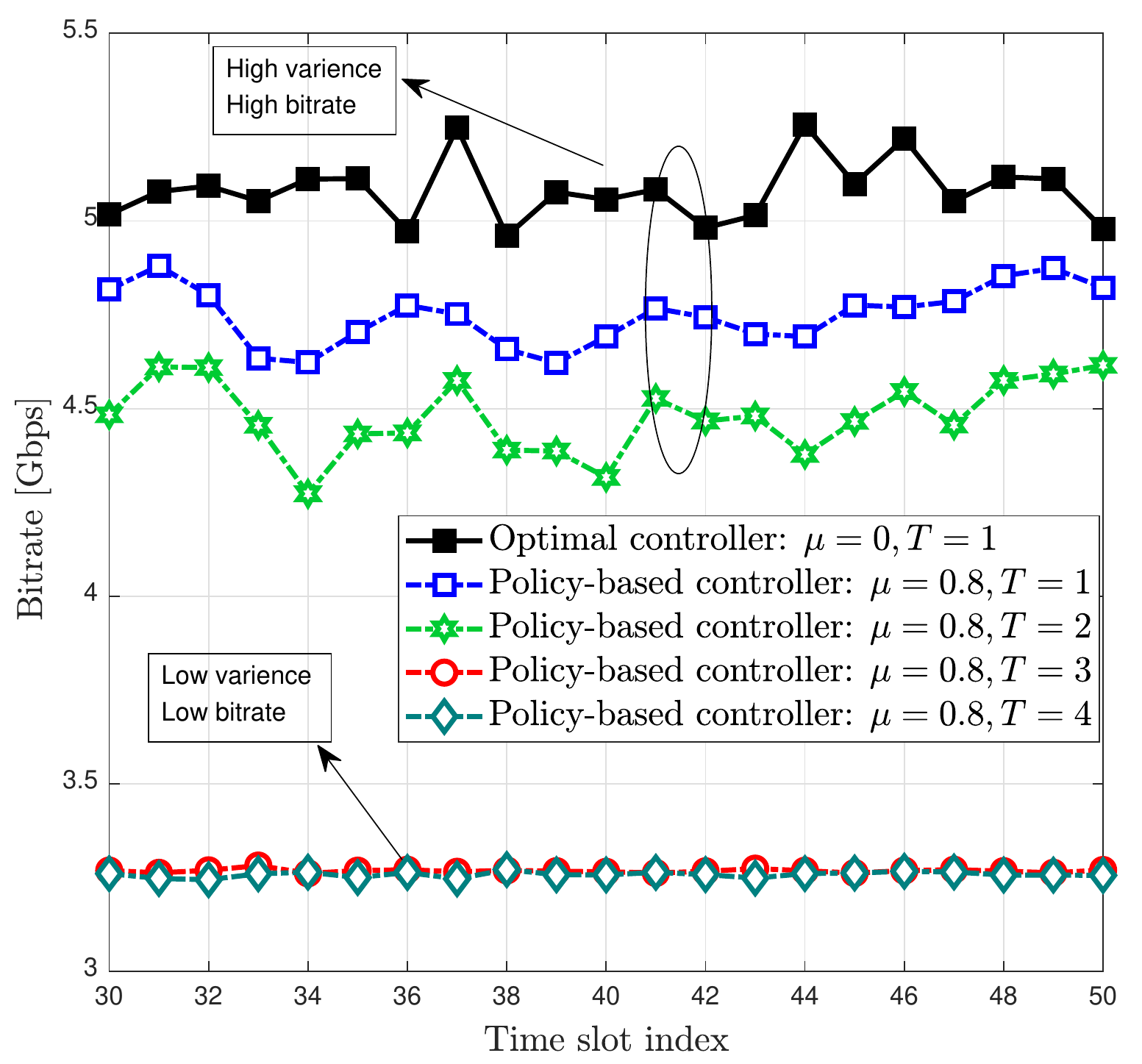}
  \caption{Centralized controller.}
  \label{Rate_TimeLab_scenario_Central}
\end{subfigure}
\caption{Achievable rate, $R_T$, for real-world dataset.}\label{Rate_Time_Lab}
\vspace{-0.1cm}
\end{figure}

In Fig.~\ref{Risk_effect}, we show the impact of the risk sensitivity parameter on the reliability of achievable rate. Indeed, in Fig.~\ref{Risk_effect}, we show the variance of received rate versus different values of the risk sensitivity parameter $\mu$ resulting from our proposed distributed RNN-based policy for real-world dataset in~\cite{dataset} and model-based dataset presented in the simulation setup. As we can see from this figure, a larger risk sensitivity parameter leads to less variance in the data rate. When we change $\mu$ from $0$ to $0.8$, the rate variance, on average, reduces $86\%$ and $54\%$ for the real-world and the model-based dataset, respectively. Moreover, the variance performance in the model-based dataset is higher than the real-world dataset because the users’ mobility trajectory in the model-based dataset is smoother than the real-world dataset.
\begin{figure}[t!]
  \centering
  \includegraphics[width=7cm]{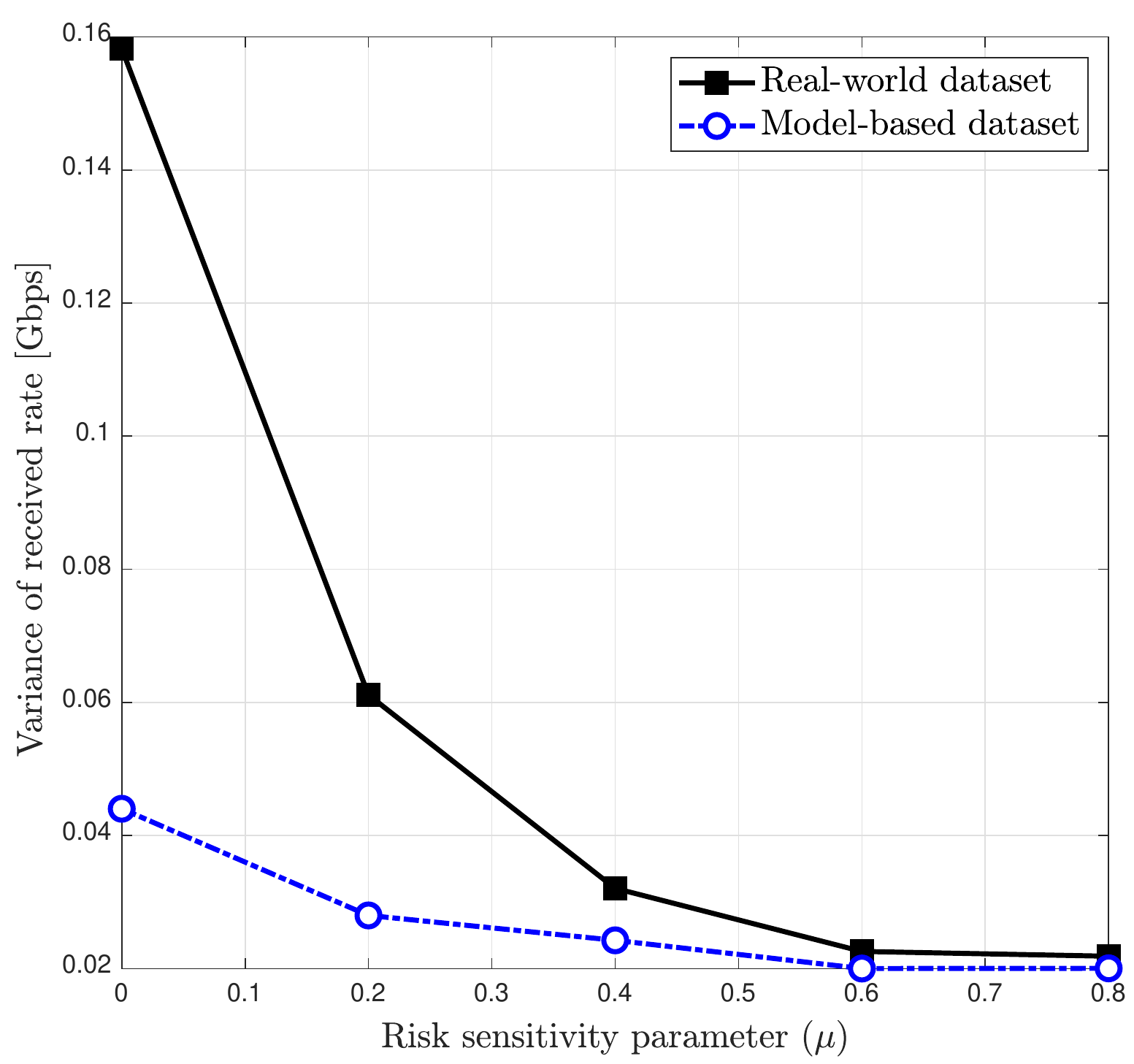}
  \caption{Impact of the risk sensitivity parameter on the achievable rate.} \vspace{-0.4cm}
  \label{Risk_effect}
  \vspace{-0.1cm}
\end{figure}

\subsection{Robustness and complexity RNN-based controllers}
Fig.~\ref{Policy} shows the average policies resulting from the centralized and distributed controllers, $\pi_{\boldsymbol{\theta}}$, and optimal joint beamforming and phase shift-controller for the mmW AP and RISs over different future consecutive time slots for the risk-sensitive approach when $\mu=0.8$. From Fig.~\ref{Policy}, the error between policies of distributed controllers and optimal solution are $1.2\%$, $2.5\%$, and $0.8\%$, for mmW AP, and RIS 1, and RIS 2 on average. This is due to the fact that during the time slots, the deep RNN, which has enough memory cell, can capture the previous dynamics over mmW links and predict the future mobile user's trajectory in a given indoor scenario. Thus, the policies from proposed phase shift-controller based on deep RNN is near the optimal solution. From Fig.~\ref{Policy}, shows that the controller steers the AP beam toward mmW RIS 1 with $-0.82$ radian and mmW RIS 2 with $-0.78$ radian with probability $0.3$ and $0.12$, respectively. Moreover, the controller of RIS 1 reflects the mmW signal from $-1.4$ to $-0.5$ radians most of the times and also the controller of RIS 2 shifts the phase of the mmW signal to cover from $-0.47$ to $-0.78$ radians with higher probability. Following the locations of mmW AP and RISs in the simulation scenario depicted in Fig.~\ref{Simulation_model}, these results are reasonable because they show that the distributed controller implemented with deep RNN coordinate the beam angle of mmW AP and phase shift-controller of RISs to cover the dark areas with high probability.
\begin{figure}[t!]
  \centering
  \includegraphics[width=7cm]{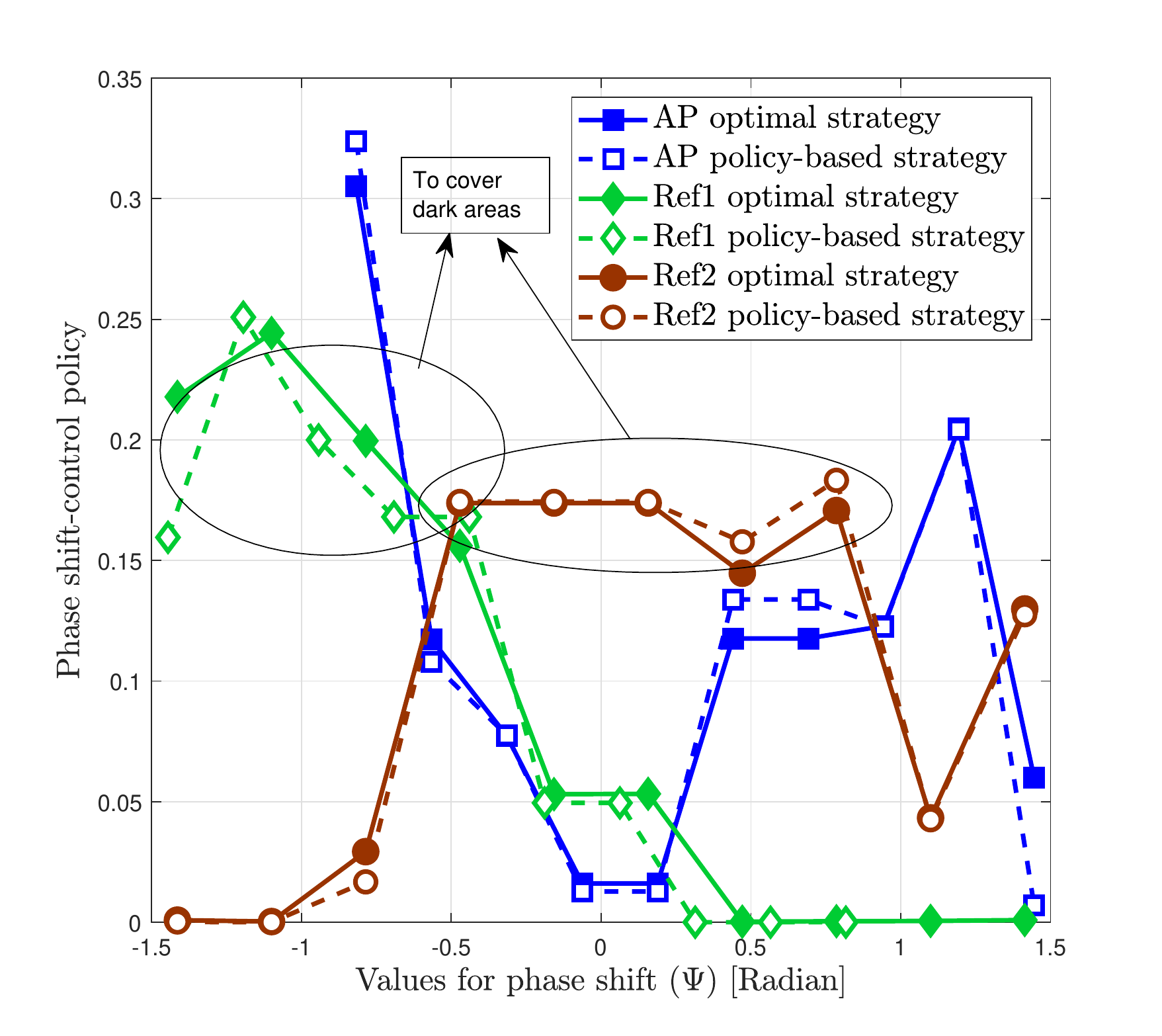}
  \caption{Optimal and policy-based strategies of joint beamforming and phase shift-controllers.} \vspace{-0.4cm}
  \label{Policy}
\end{figure}

In Fig.~\ref{Optimal_Gap}, we show, the gap between the suboptimal and optimal solutions. As we can see, the gap between the RNN-based and optimal policies for the real-world dataset is slightly different from the model-based datasets. On average, the gaps between the RNN-based and optimal policies of mmW AP and RISs are $1.7\%$ and $1.3\%$ for the real-world and the model-based dataset, respectively. Consequently, it is clear that our proposed RNN-based solution is near optimal.
\begin{figure}[t!]
  \centering
  \includegraphics[width=7cm]{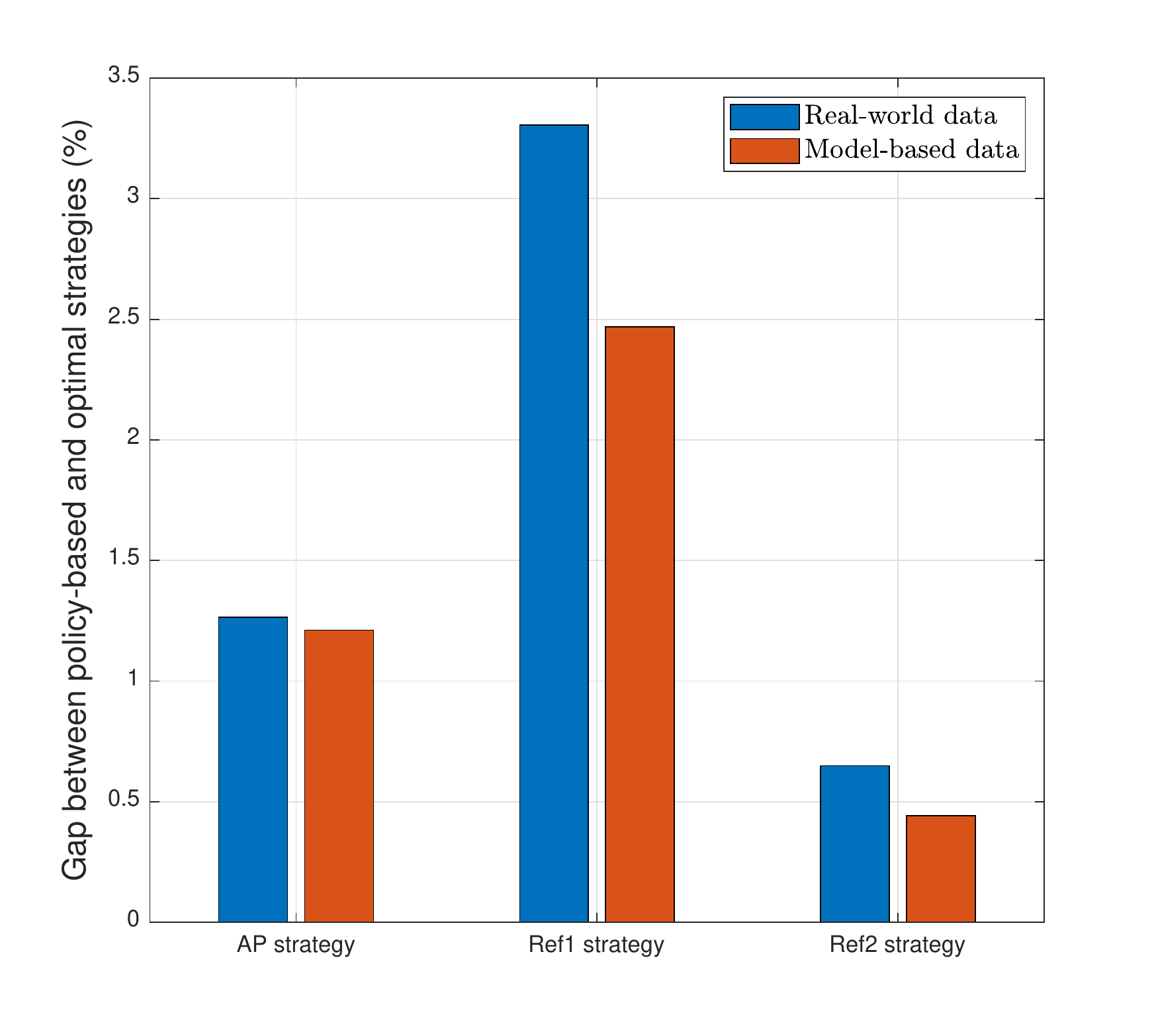}
  \caption{Gap between the RNN-based and optimal policies.}
  \vspace{-0.9cm}
  \label{Optimal_Gap}
\end{figure}

To show the robustness of our proposed scheme, we have changed the mobility pattern of the users by adding some random obstacles in the room while we use an RNN-based policy that was previously trained on a scenario without additional obstacles. This scenario allows us to evaluate the robustness of our solution with respect to new unknown random changes in the mobility pattern of users and blockages over mmW channel that were not considered in the training dataset. For this simulation, we have randomly added obstacles with size of $3\times3$ in a 35-sq. meter office environment. All the results are averaged over a large number of independent simulation runs. To evaluate the robustness of our proposed RNN-based policy, in Fig.~\ref{robustness}, we show the percentage of deviation in the data rate achieved in the new environment with respect to the scenario without additional obstacles. From Fig.~\ref{robustness}, we can see that the percentage of rate deviation increases when we add more obstacle in the room. However, when the controller predicts the policies for the next two slots, the  deviation percentage is less than $15\%$. This means our proposed control policy is more than $85\%$ robust with respect to the new environmental changes in the room. Moreover, when the RNN-based controller predicts control policy during 3 or 4 future time slots in a new environment, the robustness of our proposed RNN-based controller decreases. Hence, when $T=1$ or $2$, the RNN-based control policy, which is trained using the dataset of the previous environment, is robust enough to be used in a new environment. In contrast, when $T=3$ or $4$, we need to retrain the RNN-based control policy using the dataset of the new environment.
\begin{figure}[t!]
  \centering
  \includegraphics[width=7cm]{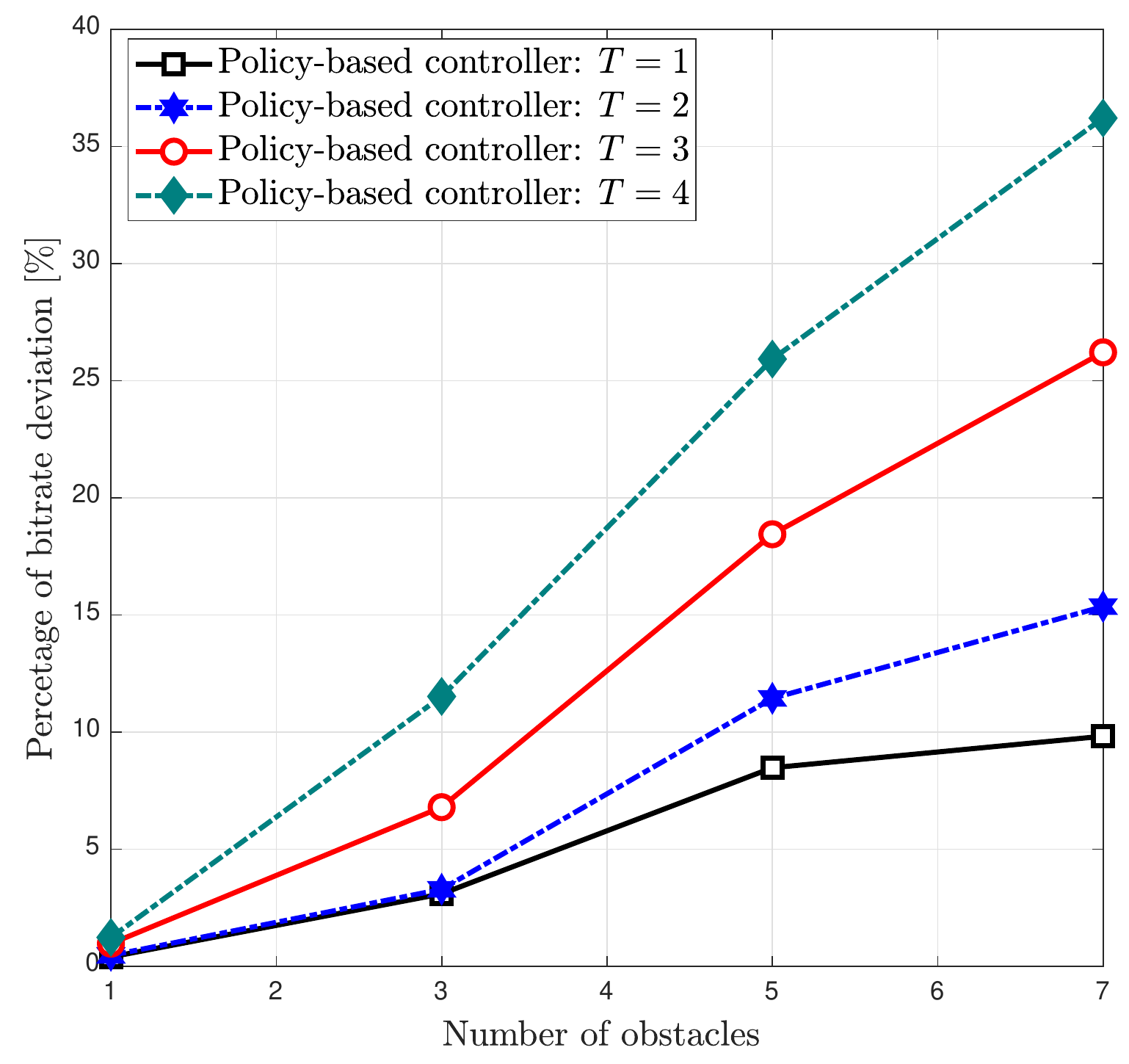}
  \caption{Percentage of rate deviation v.s. number of obstacles.}
   \vspace{-0.5cm}
  \label{robustness}
\end{figure}

\section{Conclusion}\label{Sec:Conclusion}
In this paper, we have proposed a novel framework for guaranteeing ultra-reliable mmW communications using multiple AI-enabled RISs. First, based on risk-sensitive RL, we have defined a parametric risk-sensitive episodic return to maximize the expected bitrate and mitigate the risk of  mmW link blockage. Then, we have analytically derived a closed-form approximation for the gradient of the risk-sensitive episodic return. Next, we have modeled the problem of joint beamforming for mmW AP and phase shift-controlling for mmW RISs as an identical payoff stochastic game in a cooperative multi-agent environment, in which agents are mmW AP and RISs. We have proposed two centralized and distributed controllers using deep RNNs. Then, we have trained our proposed deep RNN-based controllers based on the derived closed-form gradient of the risk-sensitive episodic return. Moreover, we have proved that the gradient updating algorithm converges to the same locally optimal parameters for deep RNN-based centralized and distributed controllers. Simulation results show that the error between policies of the optimal and proposed controllers is less than $1.5\%$. Moreover, the difference between performance of the proposed centralized and distributed controllers is less than $1\%$. On average, for high value of risk-sensitive parameter, the variance of the achievable rates resulting from deep RNN-based controllers is $60\%$ less than that of the risk-averse.
\vspace{-0.3cm}
\appendix
\section{}
\subsection{Proof of Proposition 1}
Let $\Lambda_{T}=\{(\boldsymbol{a}_{t'},r_{t'})|t=t,...,t+T-1\}$ be a trajectory during $T$-consecutive time slots which leads to the episodic reward $R_{T,t}=\sum_{t'=t}^{t+T-1}r_{t'}$. The Taylor expansion of the utility function for small values of $\mu$ yields: $J(\boldsymbol{\theta},t)\simeq \mathbb{E}_{\Lambda_{T}}\{R_{T,t}\}-\frac{\mu}{2}\text{Var}_{\Lambda_{T}}\{R_{T,t}\}$. Since $\text{Var}_{\Lambda_{T}}\{R_{T,t}\}=\mathbb{E}_{\Lambda_{T}}\{R_{T,t}^2\}-\big(\mathbb{E}_{\Lambda_{T}}\{R_{T,t}\}\big)^2$, we can rewrite:
\begin{align}
J(\boldsymbol{\theta},t)\simeq
 \mathbb{E}_{\Lambda_{T}}\{R_{T,t}-\frac{\mu}{2}R_{T,t}^2\}+
\frac{\mu}{2} \big(\mathbb{E}_{\Lambda_{T}}\{R_{T,t}\}\big)^2.
\end{align}

The probability of the trajectory $\Lambda_{T}$ is $\Pi_{\boldsymbol{\theta}}(T)$. Thus, we can write $
J(\boldsymbol{\theta},t)\simeq
 \sum_{\Lambda_{T}}\{\Pi_{\boldsymbol{\theta}}(T)(R_{T,t}-\frac{\mu}{2}R_{T,t}^2)\}+
\frac{\mu}{2} \big(\sum_{\Lambda_{T}}\Pi_{\boldsymbol{\theta}}(T)\{R_{T,t}\}\big)^2$. Hence:
\begin{align}
&\nabla_{\boldsymbol{\theta}}J(\boldsymbol{\theta},t)\simeq
\sum_{\Lambda_{T}}\{\nabla_{\boldsymbol{\theta}}\Pi_{\boldsymbol{\theta}}(T)(R_{T,t}-\frac{\mu}{2}R_{T,t}^2)\}+
\nonumber \\ &
\mu \big(\sum_{\Lambda_{T}}\nabla_{\boldsymbol{\theta}}\Pi_{\boldsymbol{\theta}}(T)\{R_{T,t}\}\big) \big(\sum_{\Lambda_{T}}\Pi_{\boldsymbol{\theta}}(T)\{R_{T,t}\}\big).
\label{Delta_J_Linear}
\end{align}

Since $\nabla_{\boldsymbol{\theta}}\log\Pi_{\boldsymbol{\theta}}(T)=\frac{\nabla_{\boldsymbol{\theta}}\Pi_{\boldsymbol{\theta}}(T)}{\Pi_{\boldsymbol{\theta}}(T)}$, we can write $\nabla_{\boldsymbol{\theta}}J(\boldsymbol{\theta},t)\approx \mathbb{E}_{\Lambda_{T}}\{\nabla_{\boldsymbol{\theta}}\log\Pi_{\boldsymbol{\theta}}(T) \big(R_{T,t}-\frac{\mu}{2}R_{T,t}^2\big)\}+\mu \mathbb{E}_{\Lambda_{T}}\{\nabla_{\boldsymbol{\theta}}\log\Pi_{\boldsymbol{\theta}}(T) R_{T,t}\}\mathbb{E}_{\Lambda_{T}}\{R_{T,t}\}$. By performing additional simplifications, we will yield (\ref{gradient_equation}).

Moreover, when the agent acts independently, the probability of trajectory $\Lambda_{T}$ is equal to $\Pi_{\boldsymbol{\theta}}(T)=\prod_{t'=t}^{t+T-1} \prod_{m \in \mathcal{M}} \pi_{\boldsymbol{\theta}_m}(a_{m,t'}|\mathcal{H}_{m,t'})\Pr\{r_{(t'+1)}|\boldsymbol{a}_{t'},\mathcal{H}_{m,t'}\}$. Due to the fact that $\log(xy)=\log(x)+\log(y)$, and $\nabla_{\boldsymbol{\theta}}\Pr\{r_{(t'+1)}|\boldsymbol{a}_{t'},\mathcal{H}_{m,t'}\}=0$, we can write $\nabla_{\boldsymbol{\theta}}\log\Pi_{\boldsymbol{\theta}}(T)=\sum_{t'=t}^{t+T-1} \sum_{m \in \mathcal{M}}$ $ \nabla_{\boldsymbol{\theta}_m} \log\pi_{\boldsymbol{\theta}_m}(a_{m,t'}|\mathcal{H}_{m,t'})$.
\subsection{Proof of Theorem 1}
Since the agents act independently, for two agents $m$ and $m'$, where $m'\neq m$, we have $\nabla_{\boldsymbol{\theta}_{m}} \log\pi_{\boldsymbol{\theta}_{m'}}(a_{m',t'}|\mathcal{H}_{t'})=0$. Thus, we can write $\nabla_{\boldsymbol{\theta}_m}\log\Pi_{\boldsymbol{\theta}}(T)=\sum_{t'=t}^{t+T-1} \nabla_{\boldsymbol{\theta}_m} \log\pi_{\boldsymbol{\theta}_m}(a_{m,t'}|\mathcal{H}_{t'})$. Then, if the agents, which are synchronized by coordinating links, act independently in a distributed manner, we have:
\begin{align}
&\nabla_{\boldsymbol{\theta}_m} J(\boldsymbol{\theta},t) \approx
 \mathbb{E}_{\Lambda_{T}}\{  \sum_{t'=t}^{t+T-1} \nabla_{\boldsymbol{\theta}_m} \log\pi_{\boldsymbol{\theta}_m}(a_{m,t'}|\mathcal{H}_{t'})  \times \label{one_ApendB}
   \nonumber \\ &
  \big((1+\mu \mathbb{E}_{\Lambda_{T}}\{R_{T,t}\}) R_{T,t}-\frac{\mu}{2}R_{T,t}^2\big)\}.
\end{align}

By comparing (\ref{one_ApendB}) and Proposition 1, we can say that (\ref{one_ApendB}) shows the results of Proposition 1 where $\nabla_{\boldsymbol{\theta}_m}\log\Pi_{\boldsymbol{\theta}}(T)=\sum_{t'=t}^{t+T-1} \nabla_{\boldsymbol{\theta}_m} \log\pi_{\boldsymbol{\theta}_m}(a_{m,t'}|\mathcal{H}_{t'})$. Whether a centralized controller is being executed by a central server, or it is implemented by agents individually executing policies synchronously, joint histories, $\mathcal{H}_t$, are generated from the same distribution $T(\boldsymbol{s}',\boldsymbol{s},\boldsymbol{a})$ and identical payoff will be achieved by mmW APs and all RISs in POIPSG. This fact shows that the distributed algorithm is sampling from the same distribution as the centralized algorithm samples. Thus, starting from the same point in the search space of policies, on the same history sequence, the gradient updating algorithm will be stepwise the same for the distributed controllers and the centralized one.
\subsection{Proof of Theorem 2}
Assume that for a given global history sequence $\mathcal{H}_{t'}$ for $t'=t, t+1, ..., t+T-1$ and at the convergence of the gradient update algorithm using (\ref{gradient_equation}), the policy profile under the distributed controllers is $\{\pi_{\boldsymbol{\theta}^*_m}|,\forall m\in \mathcal{M}\}$. At this policy profile, since all agents have an identical payoff function, the best response of agent $m$ to the given strategies of all other agents is defined as $\pi_{\boldsymbol{\theta}^{\text{b}}_m}$ where $\theta^{\text{b}}_m=\operatorname*{argmax}_{\theta_m} J(\boldsymbol{\theta}_m, \cup_{m' \in \mathcal{M} \backslash \{m\}} \boldsymbol{\theta}^*_{m'} ,t)$. In this case,  due to the fact that the agents act independently, the gradient updating rule for agent $m$ to find its best response is given by (25). Since the global history sequence $\mathcal{H}_{t'}$ for $t'=t, t+1, ..., t+T-1$ is identical for all agents, the gradient updating rule in (25) converges to $\theta_m^*$. Subsequently, based on the gradient updating rule, the best response of the agent $m$ will be $\pi_{\boldsymbol{\theta}^{\text{b}}_m}=\pi_{\boldsymbol{\theta}^{*}_m}$, if other agents choose the converged policy profiles $\boldsymbol{\theta}^*_{m'},\forall m'\neq m$. Thus, in this case, $\theta^*_m=\operatorname*{argmax}_{\theta_m} J(\boldsymbol{\theta}_m, \cup_{m' \in \mathcal{M} \backslash \{m\}} \boldsymbol{\theta}^*_{m'} ,t)$. Consequently, at the strategy profile $\{\pi_{\boldsymbol{\theta}^*_m}|,\forall m\in \mathcal{M}\}$, agent $m$ can not do better by choosing policy different from $\pi_{\boldsymbol{\theta}^*_m}$ , given that every other agent $m'\neq m$ adheres to $\pi_{\boldsymbol{\theta}^*_{m'}}$. Thus, the gradient update algorithm using (\ref{gradient_equation}) converges to the policy profile which is an NE of POIPSG under the distributed controllers.
\bibliographystyle{IEEEtran}
\def\baselinestretch{0.985}

\bibliography{Arxive}

\begin{thebibliography}{10}
\providecommand{\url}[1]{#1}
\csname url@samestyle\endcsname
\providecommand{\newblock}{\relax}
\providecommand{\bibinfo}[2]{#2}
\providecommand{\BIBentrySTDinterwordspacing}{\spaceskip=0pt\relax}
\providecommand{\BIBentryALTinterwordstretchfactor}{4}
\providecommand{\BIBentryALTinterwordspacing}{\spaceskip=\fontdimen2\font plus
\BIBentryALTinterwordstretchfactor\fontdimen3\font minus
  \fontdimen4\font\relax}
\providecommand{\BIBforeignlanguage}[2]{{%
\expandafter\ifx\csname l@#1\endcsname\relax
\typeout{** WARNING: IEEEtran.bst: No hyphenation pattern has been}%
\typeout{** loaded for the language `#1'. Using the pattern for}%
\typeout{** the default language instead.}%
\else
\language=\csname l@#1\endcsname
\fi
#2}}
\providecommand{\BIBdecl}{\relax}
\BIBdecl

\bibitem{Globecome2019}
M.~{Naderi Soorki}, W.~Saad, and M.~Bennis, ``Ultra-reliable millimeter-wave
  communications using an artificial intelligence-powered reflector,'' \emph{in
  Proc. of the IEEE Global Communications Conference (GLOBECOM)}, Waikoloa, HI,
  USA, December 2019.

\bibitem{Walid2019}
W.~{Saad}, M.~{Bennis}, and M.~{Chen}, ``A vision of {6G} wireless systems:
  Applications, trends, technologies, and open research problems,'' \emph{IEEE
  Network}, vol.~34, no.~3, pp. 134--142, 2020.

\bibitem{Petrov2017}
V.~Petrov, D.~Solomitckii, A.~Samuylov, M.~A. Lema, M.~Gapeyenko,
  D.~Moltchanov, S.~Andreev, V.~Naumov, K.~Samouylov, M.~Dohler, and
  Y.~Koucheryavy, ``Dynamic multi-connectivity performance in ultra-dense urban
  mmwave deployments,'' \emph{IEEE Journal on Selected Areas in
  Communications}, vol.~35, no.~9, pp. 2038--2055, Sep. 2017.

\bibitem{Abari2018}
O.~Abari, D.~Bharadia, A.~Duffield, and D.~Katabi, ``Enabling high-quality
  untethered virtual reality,'' \emph{in Proc. of the USENIX symposium on
  networked systems design and implementation}, pp. 1--5, Boston, MA, USA,
  March 2017.

\bibitem{Journal2019}
M.~{Naderi Soorki}, W.~{Saad}, and M.~{Bennis}, ``Optimized deployment of
  millimeter wave networks for in-venue regions with stochastic users'
  orientation,'' \emph{IEEE Transactions on Wireless Communications}, vol.~18,
  no.~11, pp. 5037--5049, Nov. 2019.

\bibitem{Chaccour2}
C.~{Chaccour}, M.~{Naderi Soorki}, W.~{Saad}, M.~{Bennis}, and P.~{Popovski},
  ``Can terahertz provide high-rate reliable low latency communications for
  wireless {VR}?'' \emph{arXiv:2005.00536}, May 2020.

\bibitem{Huang2018}
C.~{Huang}, A.~{Zappone}, G.~C. {Alexandropoulos}, M.~{Debbah}, and C.~{Yuen},
  ``Reconfigurable intelligent surfaces for energy efficiency in wireless
  communication,'' \emph{IEEE Transactions on Wireless Communications},
  vol.~18, no.~8, pp. 4157--4170, Aug 2019.

\bibitem{Christos2018}
C.~Liaskos, S.~Nie, A.~Tsioliaridou, A.~Pitsillides, S.~Ioannidis, and
  I.~Akyildiz, ``Realizing wireless communication through software-defined
  hypersurface environments,'' \emph{arXiv:1805.06677}, May 2018.

\bibitem{Basar2019}
E.~{Basar}, M.~{Di Renzo}, J.~{De Rosny}, M.~{Debbah}, M.~{Alouini}, and
  R.~{Zhang}, ``Wireless communications through reconfigurable intelligent
  surfaces,'' \emph{IEEE Access}, vol.~7, pp. 116\,753--116\,773, 2019.

\bibitem{Zhang2019}
Q.~Zhang, W.~Saad, and M.~Bennis, ``Reflections in the sky: millimeter wave
  communication with {UAV}-carried intelligent reflectors,'' \emph{in Proc. of
  the IEEE Global Communications Conference (GLOBECOM)}, Waikoloa, HI, USA,
  December 2019.

\bibitem{Chaccour1}
C.~{Chaccour}, M.~{Naderi Soorki}, W.~{Saad}, M.~{Bennis}, and P.~{Popovski},
  ``Risk-based optimization of virtual reality over terahertz reconfigurable
  intelligent surfaces,'' \emph{in Proc. of IEEE International Conference on
  Communications (ICC)}, pp. 1--6, Dublin, Ireland, June 2020.

\bibitem{Walid2000}
M.~Jung, W.~Saad, Y.~Jang, G.~Kong, and S.~Choi, ``Performance analysis of
  large intelligent surfaces {(LISs)}: asymptotic data rate and channel
  hardening effects,'' \emph{IEEE Transactions on Wireless Communications}, to
  appear, 2020.

\bibitem{58}
Q.~Wu and R.~Zhang, ``Intelligent reflecting surface enhanced wireless network
  via joint active and passive beamforming,'' \emph{IEEE Trans. Wireless
  Commun.}, vol.~18, no.~11, pp. 5394--5409, August 2019.

\bibitem{91}
------, ``Beamforming optimization for wireless network aided by intelligent
  reflecting surface with discrete phase shifts,'' \emph{IEEE Trans. Wireless
  Commun.}, vol.~68, no.~3, pp. 1838--1851, 2020.

\bibitem{102}
B.~Ning, Z.~Chen, W.~Chen, and J.~Fang, ``Beamforming optimization for
  intelligent reflecting surface assisted {MIMO}: A sum-path-gain maximization
  approach,'' \emph{IEEE Wireless Commun. Lett.}, vol.~9, no.~7, pp.
  1105--1109, March 2020.

\bibitem{9424177}
Y.~Liu, X.~Liu, X.~Mu, T.~Hou, J.~Xu, M.~D. Renzo, and N.~Al-Dhahir,
  ``Reconfigurable intelligent surfaces: principles and opportunities,''
  \emph{IEEE Communications Surveys Tutorials}, pp. 1--1, May 2021.

\bibitem{Park2019}
J.~Park, S.~Samarakoon, M.~Bennis, and M.~Debbah, ``Wireless network
  intelligence at the edge,'' \emph{IEEE Proceedings}, Nov. 2019.

\bibitem{128}
C.~Huang, G.~C. Alexandropoulos, C.~Yuen, and M.~Debbah, ``Indoor signal
  focusing with deep learning designed reconfigurable intelligent surfaces,''
  \emph{In Proc. of International Workshop on Signal Processing Advances in
  Wireless Communications (SPAWC)}, vol.~9, no.~7, pp. 1--5, Cannes, France,
  July 2019.

\bibitem{129}
J.~Gao, C.~Zhong, X.~Chen, H.~Lin, and Z.~Zhang, ``Unsupervised learning for
  passive beamforming,'' \emph{IEEE Commun. Lett.}, vol.~24, no.~5, pp.
  1052--1056, 2020.

\bibitem{144}
A.~Taha, Y.~Zhang, F.~B. Mismar, and A.~Alkhateeb, ``Deep reinforcement
  learning for intelligent reflecting surfaces: Towards standalone operation,''
  \emph{In Proc. of International Workshop on Signal Processing Advances in
  Wireless Communications (SPAWC)}, no.~5, pp. 1--5, Atlanta, GA, USA, May
  2020.

\bibitem{142}
C.~Huang, R.~Mo, and C.~Yuen, ``Reconfigurable intelligent surface assisted
  multiuser {MISO} systems exploiting deep reinforcement learning,'' \emph{IEEE
  J. Sel. Areas Commun.}, vol.~38, no.~8, pp. 1839--1850, Aug. 2020.

\bibitem{Walidbook2000}
Z.~Han, D.~Niyato, W.~Saad, and T.~{Basar}, \emph{Game theory for next
  generation wireless and communication networks: modeling, analysis, and
  design}.\hskip 1em plus 0.5em minus 0.4em\relax Cambridge University Press,
  2019.

\bibitem{Walid2017}
M.~Chen, U.~Challita, W.~Saad, C.~Yin, and M.~Debbah, ``Artificial neural
  networks-based machine learning for wireless networks: a tutorial,''
  \emph{IEEE Communications Surveys and Tutorials}, to appear, 2019.

\bibitem{Kamoda2011}
H.~Kamoda, T.~Iwasaki, J.~Tsumochi, T.~Kuki, and O.~Hashimoto, ``60-{GHz}
  electronically reconfigurable large reflectarray using single-bit phase
  shifters,'' \emph{IEEE transactions on antennas and propagation}, vol.~59,
  no.~7, pp. 2524--2531, July 2011.

\bibitem{Tan2018}
X.~Tan, Z.~Sun, D.~Koutsonikolas, and J.~M. Jornet, ``Enabling indoor mobile
  millimeter-wave networks based on smart reflect-arrays,'' \emph{in Proc. of
  IEEE Conference on Computer Communications}, pp. 270--278, Honolulu, Hi, USA,
  April 2018.

\bibitem{two}
M.~Xu, S.~Zhang, C.~Zhong, J.~Ma, and O.~A. Dobre, ``Ordinary differential
  equation-based cnn for channel extrapolation over ris-assisted
  communication,'' \emph{IEEE Communications Letters}, pp. 1--1, March 2021.

\bibitem{Shahmansoori2018}
A.~Shahmansoori, G.~E. Garcia, G.~Destino, G.~Seco-Granados, and H.~Wymeersch,
  ``Position and orientation estimation through millimeter-wave {MIMO} in {5G}
  systems,'' \emph{IEEE Transactions on Wireless Communications}, vol.~17,
  no.~3, pp. 1822--1835, March 2018.

\bibitem{ChannelEstimation}
L.~{Zhao}, D.~W.~K. {Ng}, and J.~{Yuan}, ``Multi-user precoding and channel
  estimation for hybrid millimeter wave systems,'' \emph{IEEE Journal on
  Selected Areas in Communications}, vol.~35, no.~7, pp. 1576--1590, July 2017.

\bibitem{9403420}
J.~Mirza and B.~Ali, ``Channel estimation method and phase shift design for
  reconfigurable intelligent surface assisted mimo networks,'' \emph{IEEE
  Transactions on Cognitive Communications and Networking}, pp. 1--1, April
  2021.

\bibitem{five}
M.~Bennis, M.~Debbah, and H.~V. Poor, ``Ultra-reliable and low-latency
  communication: Tail, risk and scale,'' \emph{in Proceedings of the IEEE},
  vol. 106, no. 10, pp. 1834-1853, Oct. 2018.

\bibitem{four}
O.~Mihatsch and R.~Neuneier, ``Risk-sensitive reinforcement learning,''
  \emph{Machine learning}, vol.~49, pp. 267--290, Nov. 2002.

\bibitem{Sutton2018}
R.~S. Sutton and A.~G. Barto, \emph{Reinforcement learning: an
  introduction}.\hskip 1em plus 0.5em minus 0.4em\relax second edition, MIT
  press, Cambridge, MA, 2018.

\bibitem{Polecy2000}
L.~Peshkin, K.~Kim, N.~Meuleau, and L.~Kaelbling, ``Learning to cooperate via
  policy search,'' \emph{in Proc. of the Conference on Uncertainty in
  Artificial Intelligence}, pp. 489--496, Stanford, California, USA, June-July,
  2000.

\bibitem{three}
M.~Hausknecht and P.~Stone, ``Deep recurrent {Q}-learning for partially
  observable {MDP}s,'' \emph{in Proc. of Association for the Advancement of
  Artificial Antelligence Symposium}, Arlington, Virginia, USA, July 2015.

\bibitem{Goodfellow-2016}
I.~Goodfellow, Y.~Bengio, and A.~Courville, \emph{Deep learning}.\hskip 1em
  plus 0.5em minus 0.4em\relax MIT Press, 2016,
  \url{http://www.deeplearningbook.org}.

\bibitem{9082619}
J.~Liu, J.~Chen, S.~Luo, S.~Li, and S.~Fu, ``Deep learning driven
  non-orthogonal precoding for millimeter wave communications,'' \emph{IEEE
  Journal on Emerging and Selected Topics in Circuits and Systems}, vol.~10,
  no.~2, pp. 164--176, June 2020.

\bibitem{complexity}
P.~Orponen, ``Computational complexity of neural networks: a survey,''
  \emph{Nordic Journal of Computing}, vol.~1, May 2000.

\bibitem{complexity3}
E.~Tsironi, P.~Barros, C.~Weber, and S.~Wermter, vol. 268, pp. 76--86, Dec.
  2017.

\bibitem{sevenn}
P.~F.~M. {Smulders}, ``Statistical characterization of 60-{GHz} indoor radio
  channels,'' \emph{IEEE Transactions on Antennas and Propagation}, vol.~57,
  no.~10, pp. 2820--2829, 2009.

\bibitem{dataset}
B.~T. {Morris} and M.~M. {Trivedi}, ``Trajectory learning for activity
  understanding: unsupervised, multilevel, and long-term adaptive approach,''
  \emph{IEEE Transactions on Pattern Analysis and Machine Intelligence}, pp.
  2287--2301, Nov. 2011.

\end{thebibliography}

\vspace{-0.8cm}

\begin{IEEEbiography}[{\includegraphics[width=1in,height=1.25in,clip,keepaspectratio]{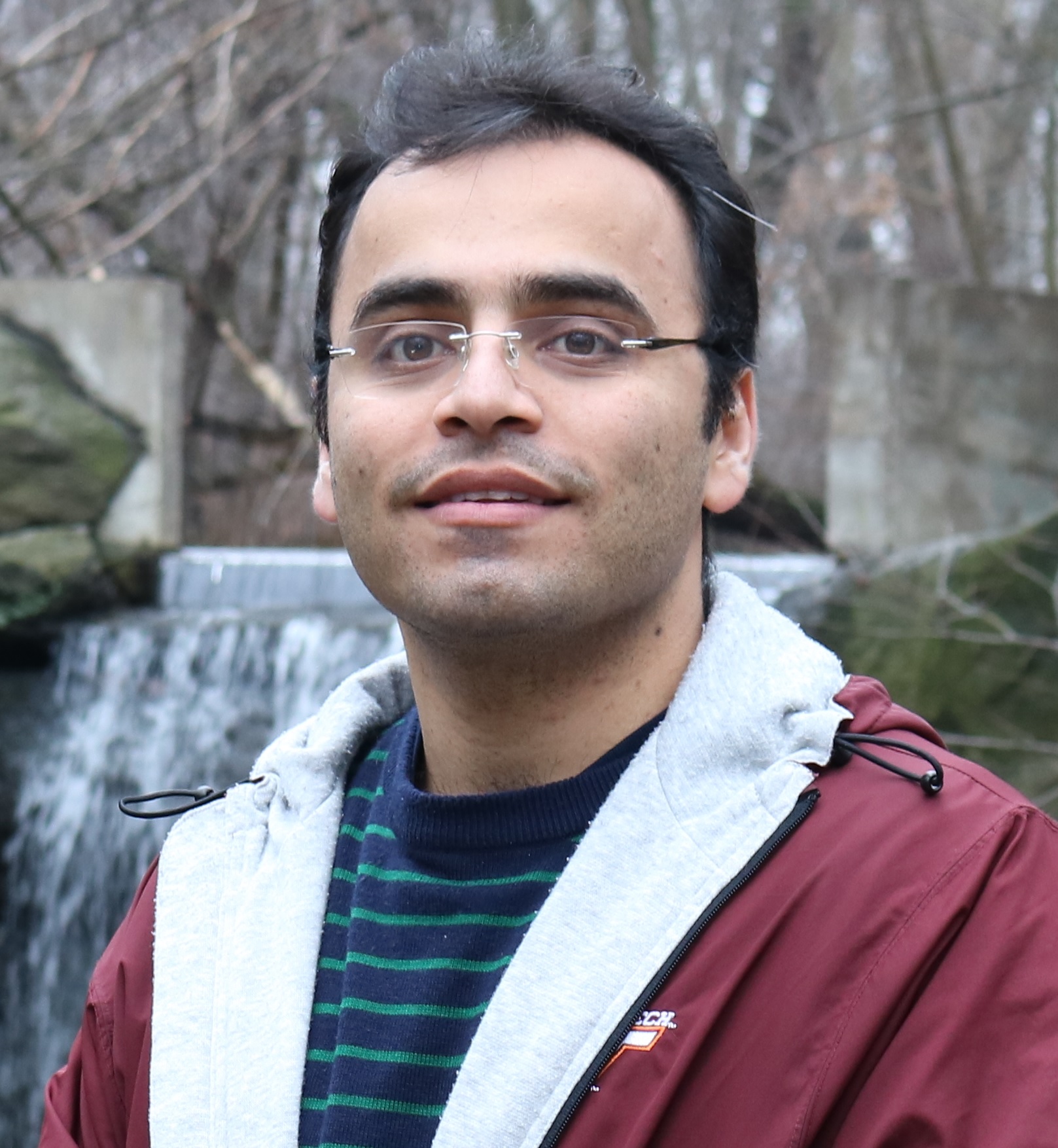}}] Mehdi~Naderi~Soorki received the BSc degree in electrical engineering from the Iran University of Science and Technology, in 2007, and the MSc and PhD degrees in telecommunication networks from the Isfahan University of Technology, in 2010 and 2018, respectively. He has been a research scholar at the Network Science, Wireless, and Security laboratory, Dept. of Electrical and Computer Engineering, Virginia Tech from 2015 to 2017. He is currently an assistant professor at electrical engineering department, engineering faculty, Shahid Chamran University of Ahvaz, Iran. His research interests include future wireless network design based on advance mathematical tools such as
game, graph, and queueing theories, optimization techniques, and machine learning.
\end{IEEEbiography}
\vfill

\vspace{-0.21cm}

\begin{IEEEbiography}[{\includegraphics[width=1in,height=1.25in,clip,keepaspectratio]{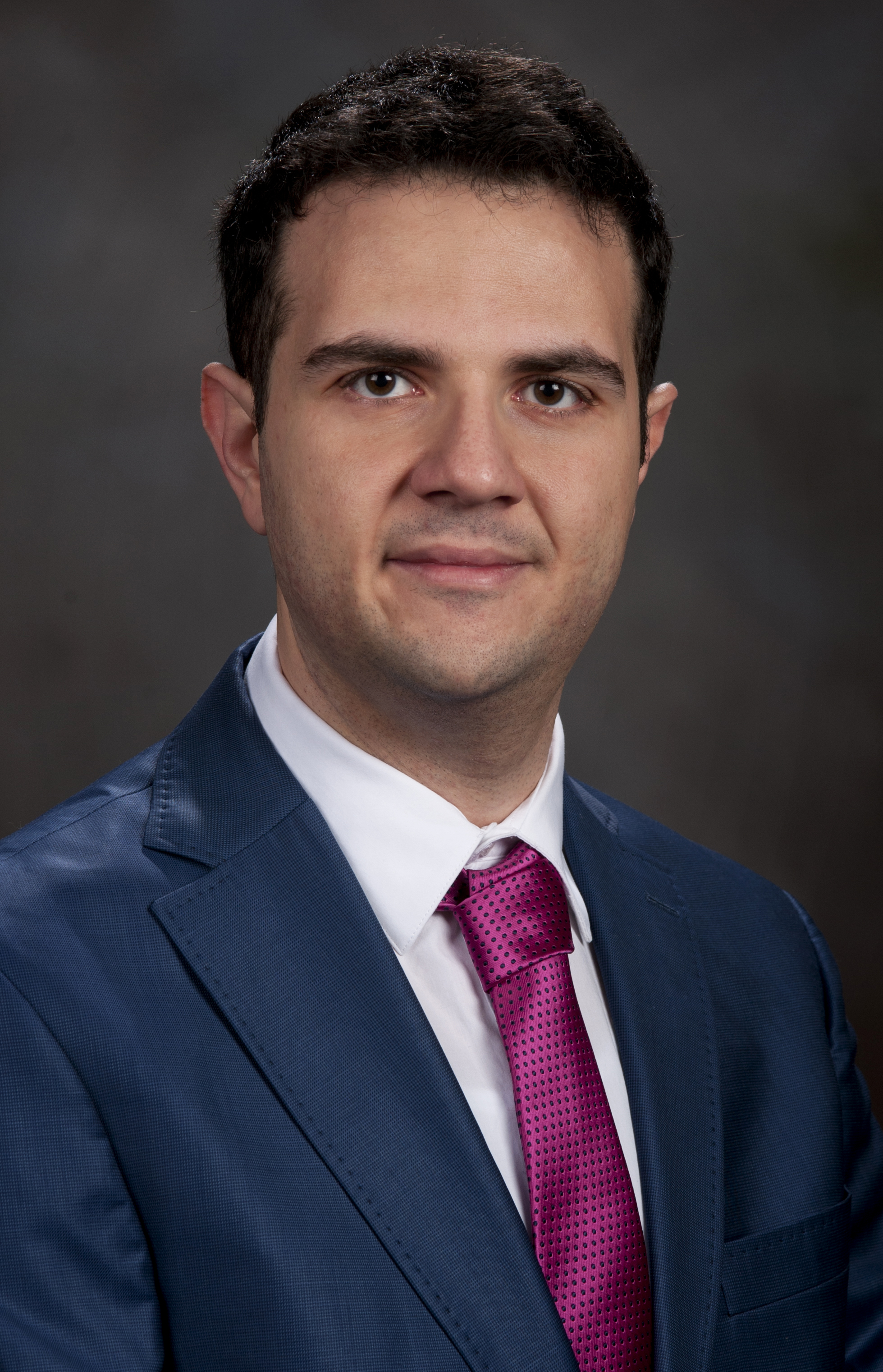}}]
Walid~Saad (S’07, M’10, SM’15, F’19) received his Ph.D degree from the University of Oslo in 2010. He is currently a Professor at the Department of Electrical and Computer Engineering at Virginia Tech, where he leads the Network sciEnce, Wireless, and Security (NEWS) laboratory. His research interests include wireless networks, machine learning, game theory, security, unmanned aerial vehicles, cyber-physical systems, and network science. Dr. Saad is a Fellow of the IEEE and an IEEE Distinguished Lecturer. He is also the recipient of the NSF CAREER award in 2013, the AFOSR summer faculty fellowship in 2014, and the Young Investigator Award from the Office of Naval Research (ONR) in 2015. He was the author/co-author of ten conference best paper awards at WiOpt in 2009, ICIMP in 2010, IEEE WCNC in 2012, IEEE PIMRC in 2015, IEEE SmartGridComm in 2015, EuCNC in 2017, IEEE GLOBECOM in 2018, IFIP NTMS in 2019, IEEE ICC in 2020, and IEEE GLOBECOM in 2020. He is the recipient of the 2015 Fred W. Ellersick Prize from the IEEE Communications Society, of the 2017 IEEE ComSoc Best Young Professional in Academia award, of the 2018 IEEE ComSoc Radio Communications Committee Early Achievement Award, and of the 2019 IEEE ComSoc Communication Theory Technical Committee. He was also a co-author of the 2019 IEEE Communications Society Young Author Best Paper and of the 2021 IEEE Communications Society Young Author Best Paper. From 2015-2017, Dr. Saad was named the Stephen O. Lane Junior Faculty Fellow at Virginia Tech and, in 2017, he was named College of Engineering Faculty Fellow. He received the Dean's award for Research Excellence from Virginia Tech in 2019. He currently serves as an editor for the IEEE Transactions on Mobile Computing, the IEEE Transactions on Network Science and Engineering, and the IEEE Transactions on Cognitive Communications and Networking. He is an Editor-at-Large for the IEEE Transactions on Communications.
\end{IEEEbiography}
\vfill

\begin{IEEEbiography}[{\includegraphics[width=1in,height=1.25in,clip,keepaspectratio]{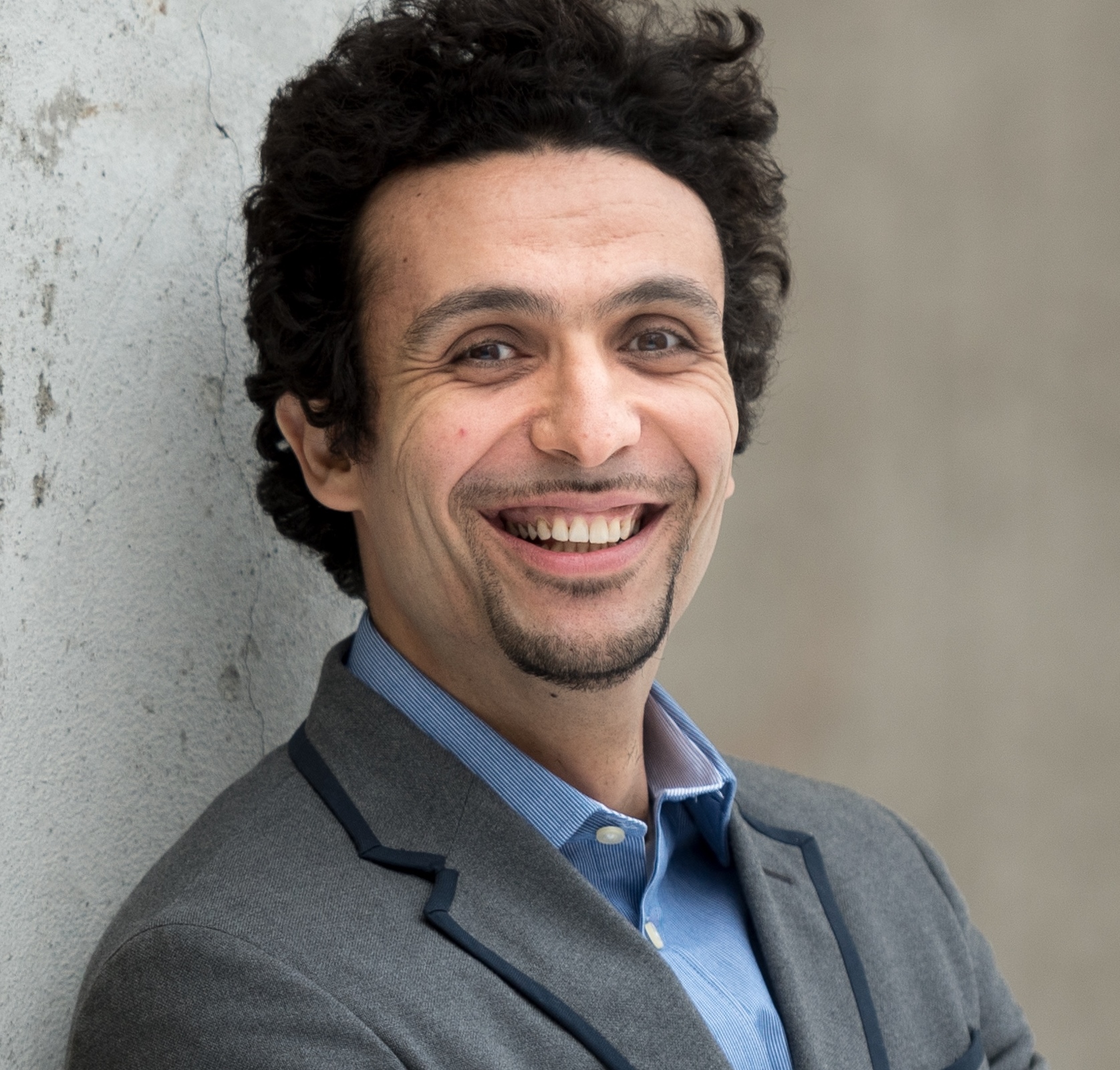}}]
Mehdi~Bennis is an Associate Professor at the Centre for Wireless Communications, University of Oulu, Finland, an Academy of Finland Research Fellow and head of the intelligent connectivity and networks/systems group (ICON). His main research interests are in radio resource management, heterogeneous networks, game theory and machine learning in 5G networks and beyond. He has co-authored one book and published more than 200 research papers in international conferences, journals and book chapters. He has been the recipient of several prestigious awards including the 2015 Fred W. Ellersick Prize from the IEEE Communications Society, the 2016 Best Tutorial Prize from the IEEE Communications Society, the 2017 EURASIP Best paper Award for the Journal of Wireless Communications and Networks, the all-University of Oulu award for research and the 2019 IEEE ComSoc Radio Communications Committee Early Achievement Award. Dr Bennis is an editor of IEEE TCOM.
\end{IEEEbiography}

\begin{IEEEbiography}[{\includegraphics[width=1in,height=1.25in,clip,keepaspectratio]{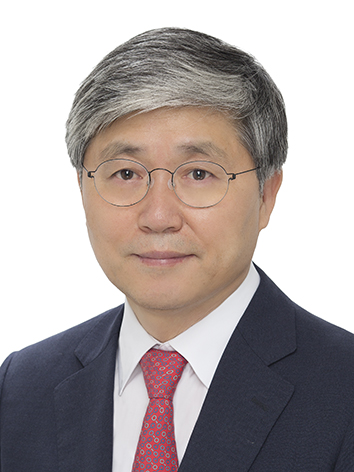}}]
Choong~Seon~Hong (S’95-M’97-SM’11) received the B.S. and M.S. degrees in electronic engineering from Kyung Hee University, Seoul, South Korea, in 1983 and 1985, respectively, and the Ph.D. degree from Keio University, Tokyo, Japan, in 1997. In 1988, he joined KT, Gyeonggi-do, South Korea, where he was involved in broadband networks as a member of the Technical Staff. Since 1993, he has been with Keio University. He was with the Telecommunications Network Laboratory, KT, as a Senior Member of Technical Staff and as the Director of the Networking Research Team until 1999. Since 1999, he has been a Professor with the Department of Computer Science and Engineering, Kyung Hee University. His research interests include future Internet, intelligent edge computing, network management, and network security. Dr. Hong is a member of the Association for Computing Machinery (ACM), the Institute of Electronics, Information and Communication Engineers (IEICE), the Information Processing Society of Japan (IPSJ), the Korean Institute of Information Scientists and Engineers (KIISE), the Korean Institute of Communications and Information Sciences (KICS), the Korean Information Processing Society (KIPS), and the Open Standards and ICT Association (OSIA). He has served as the General Chair, the TPC Chair/Member, or an Organizing Committee Member of international conferences, such as the Network Operations and Management Symposium (NOMS), International Symposium on Integrated Network Management (IM), Asia-Pacific Network Operations and Management Symposium (APNOMS), End-to-End Monitoring Techniques and Services (E2EMON), IEEE Consumer Communications and Networking Conference (CCNC), Assurance in Distributed Systems and Networks (ADSN), International Conference on Parallel Processing (ICPP), Data Integration and Mining (DIM), World Conference on Information Security Applications (WISA), Broadband Convergence Network (BcN), Telecommunication Information Networking Architecture (TINA), International Symposium on Applications and the Internet (SAINT), and International Conference on Information Networking (ICOIN). He was an Associate Editor of the IEEE TRANSACTIONS ON NETWORK AND SERVICE MANAGEMENT and the IEEE JOURNAL OF COMMUNICATIONS AND NETWORKS. He currently serves as an Associate Editor for the International Journal of Network Management.
\end{IEEEbiography}

\vfill
\balance
\end{document}